\documentclass[11pt]{article}

\usepackage{authblk}

\usepackage[pdftex]{graphicx}
\usepackage[ddmmyyyy]{datetime}

\usepackage{amssymb,amsfonts,amsmath}
\usepackage{amsthm,enumerate}
\usepackage{mathtools}
\usepackage{fancyhdr}
\usepackage{hyperref}
\usepackage[sort,compress,numbers]{natbib}
\input xy
\xyoption{all}

\usepackage[version=3]{mhchem}
\usepackage{color}

\newcommand{\N}{\mathbb{N}}

\newcommand{\R}{\mathbb{R}}

\newcommand{\mmS}{\mathcal{S}}
\newcommand{\mmY}{\mathcal{Y}}

\newcommand{\mmC}{\mathcal{C}}

\newcommand{\mmR}{\mathcal{R}}

\newcommand{\omR}{\overline{\R}}

\DeclareMathOperator*{\im}{Im}
\DeclareMathOperator*{\Con}{Con}

\renewcommand{\Re}{{\rm Re}}
\newcommand{\CC}{\mathbb C}

\newcommand{\adj}{\text{adj}}

\newcommand{\sign}{\text{sign}}

\def\wk{\widetilde{k}}
\def\tot{\text{cons}}

\def\C{\text{C}}
\def\B{\text{B}}
\def\A{\text{A}}

\newtheorem{theorem}{Theorem}

\newtheorem{proposition}{Proposition}
\newtheorem{corollary}{Corollary}
\newtheorem{lemma}{Lemma}

\theoremstyle{definition}
\newtheorem{definition}[equation]{Definition}
\newtheorem{remark}[equation]{Remark}

\begin{document}

\thispagestyle{plain}

\begin{center}
{\Large \bf Simplifying Biochemical Models With Intermediate Species}

\bigskip
Elisenda Feliu$^1$, Carsten Wiuf$^{1}$

\footnotetext[1]{Department of Mathematical Sciences, University of Copenhagen, 
Universitetsparken 5, 2100 Copenhagen, Denmark.
E-mail:  efeliu@math.ku.dk, wiuf@math.ku.dk.}
\date{\today}
\end{center}

\begin{abstract} 
Mathematical models are increasingly being used to understand complex biochemical systems, to analyze experimental data and make predictions about unobserved quantities. However,  we rarely know how robust our conclusions are with respect to the choice and uncertainties of the model.
Using algebraic techniques we study systematically the effects of intermediate, or transient, species in biochemical systems and provide a simple, yet rigorous mathematical classification of all models obtained from a \emph{core model} by including intermediates. Main examples include enzymatic and post-translational modification systems, where intermediates often are considered insignificant and neglected in a model, or they are not included because we are unaware of their existence.  
All possible models obtained from the core model are classified into a finite number of classes. Each class is defined by  a mathematically simple \emph{canonical model} 
that characterizes crucial dynamical  properties, such as  mono- and multistationarity and  stability  of  steady states, of all models in the class. We show that if the core model does not have conservation laws, then the introduction of intermediates does not change the steady-state concentrations of the species in the core model, after suitable matching of parameters.
Importantly, our results provide guidelines to  the modeler in choosing between models and in distinguishing their properties. Further, our work provides a formal way of comparing models that share a  common skeleton.

\end{abstract}

\section*{Introduction}
Systems biology aims to understand complex systems and  to build mathematical models 
that are useful for inference and prediction. However, model building is rarely straightforward and we typically seek a compromise between the simple and the accurate, shaped by our current knowledge of the system. 
Two models of the same system, potentially differing in  the number of species and the form of reactions, 
might have different qualitative properties and the conclusions we draw from analyzing the models might be strongly model dependent. The predictive value and biological validity of the conclusions might thus be questioned. 
It is therefore important to understand the role and consequences of model choice and model uncertainty in modeling biochemical systems.

Transient, or intermediate, species in biochemical reaction pathways are often ignored in models or  grouped into a single or few components, either  for reasons of simplicity or conceptual clarification, or because of lack of knowledge. For example, models of the multiple phosphorylation systems vary considerably in the details of intermediates \cite{Chan:2012et,Markevich-mapk} and  intermediates are often ignored in models of 
 phosphorelays and two-component systems \cite{Kim:2006p77,CsikaszNagy:2011p494}. Typically, intermediate species are protein complexes such as a kinase-substrate protein complex. It has been shown that sequestration of intermediates can cause ultrasensitive behavior in some systems (e.g. \cite{Legewie:2005hw,Ventura-Hidden}). Therefore, the  inclusion/exclusion of intermediates is a matter of considerable concern.

 As an example, consider  the transfer of a modifier molecule, such as a phosphate group in a two-component system, from one molecule to another: $A^*+B\rightleftharpoons \ldots \rightleftharpoons A+B^*,$ where $A,B$ are unmodified forms (without the modifier group), $A^*,B^*$ are modified forms (with the modifier),  $\rightleftharpoons$ indicate reversible reactions, and $\ldots$ are potential transient reaction steps. Two-component systems are ubiquitous in nature and vary considerably  in architecture and mechanistic details across species and functionality \cite{krell}.  Whether or not the specifics are known beforehand, it is custom to use a reduced scheme such as $A^*+B\rightleftharpoons A+B^*$ \cite{Kim:2006p77,CsikaszNagy:2011p494}.

We use \emph{Chemical Reaction Network Theory} (CRNT) to model a  system of biochemical reactions and assume that the reaction rates follow mass-action kinetics. The  polynomial form of the reaction rates have made it possible to apply algebraic techniques to learn about qualitative properties of models, without resorting to numerical approaches  \cite{TG-nature,shinar-science,Karp:2012hz,harrington-model,Feliu:2010p94,Feliu:royal,harrington-feliu}.
Building on  previous work \cite{king-altman,TG-rational,Fel_elim}, we propose a mathematical framework  to compare different models and to study the dynamical properties of models that differ in how intermediates are included. The most fundamental and crucial dynamical features are the number and stability of steady states. We assume that the kinetic parameters are unknown and  study the capacity of each model to exhibit  different steady-state features. 

The paper is organized in the following way. We first introduce the concepts of a core model and an extension model. An extension model is constructed from the core model by  including intermediates. Next, we discuss how the steady-state equations of   different models are related and illustrate the findings with an example. We proceed to discuss the number of steady states of core and extension models. After that we introduce the steady-state classes, a key concept of this paper. Extension models in the same steady-state classes have the same properties at steady-state (provided the parameter sets of the two models can be matched, in some sense). Using these ideas, we build a decision tree to guide the modeler in choosing a model and in understanding the consequences of choosing a particular model. Finally, we illustrate our approach with an example based on two-component systems. All proofs and mathematical details are in the appendix.

\section{The core model and its extensions}

We use the notation and formalism of CRNT (see for example \cite{feinbergnotes,gunawardena-notes}).  A \emph{reaction network} is defined as a set species, denoted by capital letters (for example, $A, B, C$),  a set of complexes and a set of reactions between complexes. Each complex is a combination of species, for example $y_1=A+B$ or $y_2=2C$ (not to be confused with a protein complex). A potential reaction could be $A+B\to 2C$, or also written simply $y_1\to y_2$.  A reaction is not necessarily reversible, that is, we can have $A+B\to 2C$ without having the reverse reaction $2C\to A+B$. Whenever a reaction is reversible we model it as two separate irreversible reactions. 
We assume that each reaction occurs according to mass-action kinetics, that is, at a rate proportional to the product of the species concentrations in the reactant or source complex \cite{enz-kinetics}. For example, the reaction $A+B\to 2C$ occurs at a rate $k [A][B]$, where $[A], [B]$ are the concentrations of the species $A, B$ and $k$ is a reaction specific positive constant.
Reaction networks are often drawn graphically as in Figs.\,1A-E.  Figs.\,1C-E are schematic representations of reaction networks: only the structure of the network is shown and neither the species nor the rate constants are indicated.

Fig.\,1A  corresponds to a simple enzymatic mechanism where $E$  is an enzyme and $S_j$ is a substrate with $j=0,1,2$ phosphorylated sites. The substrate $S_0$ can be doubly phosphorylated  sequentially via $S_1$ or directly (processively). In Fig.\,1B,  a transient product $Y$ formed by $S_0$ and $E$, or by $S_1$ and $E$ (these are often denoted by $S_0\cdot E$ and $S_1\cdot E$) is shown. In the particular case we do not distinguish between the two transient products (which might be unrealistic, but it serves an illustrative purpose).

\begin{figure}[!t]
\centering
\includegraphics{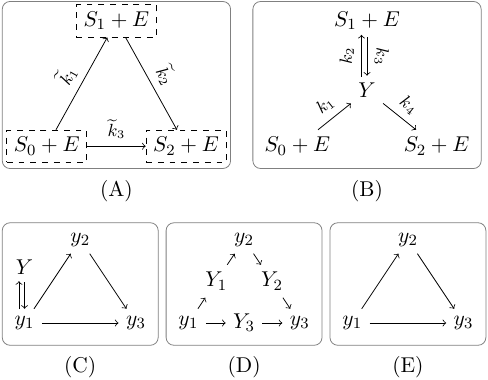}
\caption{\small Representation of reaction networks: (A)-(B) detailed representation; (C)-(E) schematic representation. (A) and (E) are core models and (B)-(D) are extended models of (A) and (E).
{\bf (A)}  A reaction network with complexes $S_0+E,S_1+E,S_2+E$ (enclosed in dashed boxes). Each reaction is labelled with its rate constant ($k$ or $\tilde{k}$).
{\bf (B)} An extension model of network (A) with intermediate $Y$.
{\bf (C)} The complex $y_1$ is involved in a reversible  ``dead-end" reaction with one intermediate. {\bf (D)} The complex $y_1$ is converted into $Y$, which splits  into $y_2$ or $y_3$, respectively (the former reversibly). {\bf (E)}  Schematic representation of (A).    
 }
\end{figure}

An \emph{intermediate} is defined as a species in a reaction network that is created and dissociated in isolation, that is, it is produced in at least one reaction, consumed in at least one reaction and  it cannot be part of any other complex (for example, $Y$ in Figs.\,1B-D). A \emph{core model} is the minimal reaction mechanism to be modeled. Each reaction $y_i\rightarrow y_j$ in the core model consists of two {\it core complexes} $y_i,y_j$.   The species contributing to the core complexes are referred to as {\it core species}. 
An \emph{extension model}  is any reaction network such that: 
\begin{itemize}
\item[(i)] The set of complexes consists of core complexes and some intermediates that are not part of the core model. 
\item[(ii)] Reactions are between two core complexes, two intermediates or between an intermediate and a core complex.
\item[(iii)]  The core model is obtained from the extension model by \textit{collapsing} all reaction paths $y_i\rightarrow Y_1\rightarrow\ldots\rightarrow Y_k\rightarrow y_j$, where $Y_i$ are intermediates, into a single reaction $y_i\rightarrow y_j$.
\end{itemize}
Some examples are given in Fig.\,1. 
Fig.\,1B is an extension model of Fig.\,1A and Figs.\,1C,D are extension models of Fig.\,1E.  Fig.\,1A is a concretization of Fig.\,1E.
 Observe that the directionality of the reaction arrows needs to be preserved. 
For instance, in Fig.\,1E, an extension of the reaction $y_1\rightarrow y_2$ cannot be $y_1\rightleftharpoons Y \rightleftharpoons y_2$, because it would imply that  $y_2\rightarrow y_1$ also is in the core model.
By adding arbitrarily many intermediates (e.g.\ $y_1\rightleftharpoons Y_1\rightleftharpoons\ldots \rightleftharpoons Y_k$) we can create arbitrarily many extension models with the same core.

Under mass-action kinetics, the dynamics of Fig.\,1B is described by a polynomial system of ordinary differential equations (ODEs): 
\begin{align}\label{odes}
\dot{[S_0]} & =  -k_1[S_0][E] ,\nonumber \\
\dot{[S_1]} & =  -k_3[S_1][E] + k_2[Y], \nonumber \\
\dot{[S_2]} & =  k_4 [Y],\nonumber \\
 \dot{[E]} & = -k_1[S_0][E]    -k_3[S_1][E] + k_2[Y]  + k_4[Y],\\
\dot{[Y]} &=k_1[S_0][E]+k_3[S_1][E]-k_2[Y]-k_4[Y] ,\nonumber 
\end{align}
where $k_*$ are rate constants, $[X]$ denotes the concentration of species $X$,   and $\dot{[X]}$ is the instantaneous change in $[X]$. 
In addition there are two \textit{conservation laws},
\begin{equation}\label{consw}
S^\B_{\tot}= [S_0]+[S_1]+[S_2]+[Y],\quad E^\B_{\tot}= [E]+[Y],
\end{equation}
that is, quantities that are conserved over time and determined by the initial concentrations.
Conservation laws confine the dynamics to an invariant space given by $S^\B_{\tot}$ and $E^\B_{\tot}$ (referred to as \emph{conserved amounts}), and the dynamical ana\-ly\-sis must be restricted to this space. The invariant spaces are called \emph{stoichiometric classes} in the CRNT literature. If we consider a maximal set of independent conservation laws, then the species that appear in the conservation laws are independent of the chosen set.

The core model in Fig.\,1A has two conservation laws,
\begin{equation}\label{consw2}
S^\A_{\tot} = [S_0]+[S_1]+[S_2],\quad \text{ and }\quad E^\A_{\tot} = [E].
\end{equation}
The two sets of conservation laws, \eqref{consw} and \eqref{consw2}, differ by a linear combination of intermediate concentrations (here a single term). This similarity between \eqref{consw} and \eqref{consw2} holds generally:

\medskip
\noindent
{\bf Theorem 1:} The conservation laws in the core model are in one-to-one correspondence with the conservation laws in any extension model. The correspondence is obtained by adding a suitable linear combination of the  $[Y]$'s to each conservation law     of the core model. \hfill{$\qed$}

\medskip
The theorem does not depend on the assumption of mass-action kinetics but relies on the structure of the network only, that is, on the set of reactions of the network.

\section{Steady-state equations}

We next state  two theorems  that allow us to relate the dynamics near steady states of the core and extension models to each other. 

At steady state  $\dot{[X]}=0$ for all species $X$.
Under the assumption of mass-action kinetics, this condition translates into a system of polynomial equations in the species concentrations. A way to solve the equations is to express one variable in terms of other variables. This expression must then be satisfied by any  solution to the system.
We let  $[y]$  denote the product of the species concentrations in  complex $y$, for example, $[2S] = [S]^2$ and $[S_0+E]=[S_0][E]$. 
Different extension models contain different intermediates, resulting in different steady-state equations. Since the intermediates always appear as linear terms in the steady-state equations of an extension model (see for example \eqref{odes}), they can  be eliminated from the equations and written in terms of the concentrations of the core species:

\medskip
\noindent
{\bf Theorem 2} \cite{fwptm,Fel_elim}: Using the equations $[\dot{Y}]=0$ for all intermediates in the extension model, the steady-state concentrations  of the intermediates $Y$ are given as  linear sums $[Y]=\sum_y \mu_{Y,y} [y]$ of products of the core species concentration. The constant $\mu_{Y,y}$ is either zero or positive and  depends only on the rate constants of the extension model.    $[y]$ appears in the expression, that is, $\mu_{Y,y}\neq  0$, if and only if there is a reaction path $y\rightarrow  \ldots\rightarrow Y$ involving  exclusively intermediates. 
 \hfill{$\qed$}

\medskip
As a consequence of the theorem, once the steady-state concentrations of the core species are known, the steady-state concentrations of the intermediates are also known. Because $\mu_{Y,y}\geq  0$ and at least one of the constants is non-zero (all intermediates are produced), positive steady-state concentrations of the core species lead to positive concentrations of the intermediates.

 The theorem makes explicit use of mass-action kinetics. It remains true for non-mass action kinetics in the sense that an explicit expression for $[Y]$ can be found if all reactions $Y\to y'$ have mass-action reaction rates, whereas all other reactions can have arbitrary reaction rates. In that case, however, the form of the expression might not be polynomial nor lead to positive concentrations.

The manipulations leading to the expression $[Y]=\sum_y \mu_{Y,y} [y]$ from $[\dot{Y}]=0$ are purely algebraic and do not require any assumptions about the  conserved amounts.
In  example  \eqref{odes}, the equation $\dot{[Y]}=0$ gives
\begin{align}\label{Ylinear}
[Y] & =m_1 [S_0][E] +m_3 [S_1][E],
\end{align}
where $m_i=\frac{k_i}{k_2+k_4}$  are reciprocal Michaelis-Menten constants \cite{enz-kinetics}. If  \eqref{Ylinear} is  substituted into \eqref{odes}, we obtain a new ODE system:
\begin{align}\label{odes2}
\dot{[S_0]} & =  -k_1[S_0][E], \nonumber\\
 \dot{[S_1]} & = -k_4m_3 [S_1][E]+ k_2m_1 [S_0][E], \nonumber \\ 
 \dot{[S_2]} & =  k_4m_1 [S_0][E] + k_4m_3 [S_1][E], \\
  \dot{[E]} & = 0, \nonumber
\end{align}
which is a mass-action system for the core model in Fig.\,1A with $\wk_1=k_2m_1$, $\wk_2=k_4m_3$,  and $\wk_3 = k_4m_1$  (as $k_1=\wk_1+\wk_3=k_2m_1+ k_4m_1$).  We say that the rate constants $\wk_*$ are \emph{realized} by $k_*$ and that $k_*$ and $\wk_*$ are a pair of \emph{matching} rate constants.
In the particular case,  $\wk_1,\wk_2,\wk_3$ are realized by choosing $k_1= \wk_1+\wk_3$,  $k_3=(\wk_1+\wk_3)\wk_2/\wk_3$ and any $k_2,k_4$ such that $k_4=k_2 \wk_3/\wk_1$. Choosing $k_2$ fixes the  values of $m_1,m_3$ in \eqref{Ylinear}. However, for some (unrealistic) extension models, not all choices of rate constants of the core model are realizable  (see appendix).

The relation between the ODEs in Fig.\,1A and 1B holds generally for any pair of core and extension models:

\medskip
\noindent
{\bf Theorem 3:} After substituting the expressions $[Y]=\sum_y \mu_{Y,y} [y]$ into the ODEs of the extension model, we obtain a mass-action system for the core model.  \hfill{$\qed$}

\medskip
The quasi-steady-state approximation (QSSA)  proceeds similarly \cite{segal}. An equation of the form $[\dot{Y}]=0$  is used to find an expression for  $[Y]$ in terms of $[y]$ under the additional assumptions that certain species are in high or low concentration. This expression is subsequently  substituted into the remaining ODE equations to reduce the system. Theorems~2 and 3 show that this always can be done, irrespectively of any biological justification of the procedure.

As a consequence of the theorems, the steady-states of  an extension model can be found in this way:
We first solve the equations $[\dot{Y}]=0$  for $[Y]$ in terms of $[y]$ (Theorem~2) and then insert the expressions for $[Y]$ into the remaining steady-state equations (Theorem~3).  
The steady states of the extended model are now found by solving the steady-state equations for the core model to obtain the concentrations of the core species. This corresponds to solve  \eqref{odes2} in the example above. The obtained values are subsequently plugged into the expressions given in Theorem 2 to find the steady-state values of the intermediates. 
That is, for matching rate constants between the core and an extension model, the solutions to the steady-state equations of the core model completely determine the solutions to the steady-state equations of the extension model. 

The conservation laws, however, impose different constraints on the steady-state solutions for given conserved amounts. 
Specifically, by inserting \eqref{Ylinear} into \eqref{consw} we obtain
\begin{align}\label{cons_law}
S^\B_{\tot} & =[S_0]+[S_1]+[S_2]+m_1 [S_0][E] + m_3[S_1][E], \nonumber \\
 E^\B_{\tot} & = [E]+ m_1 [S_0][E] + m_3 [S_1][E].
\end{align} 
The steady states of the extension model  solve  \eqref{odes2} and  \eqref{cons_law}, while they solve \eqref{odes2} and  \eqref{consw2} in the core model. Equation~\eqref{cons_law} is non-linear in the concentrations of the core species. Non-linear terms in the conservation laws can cause the two models to have substantially different properties. This is reflected in the example in the next section.

 Importantly, if the system has no conservation law, then addition of intermediates cannot alter any property of the core model at steady state. This will be the case, for instance, when production and degradation of all core species in the model are explicitly modeled.

\section{An example}

The number of steady-state  solutions for the core model and an extension model can differ substantially. For matching rate constants, the steady states of each system are found by intersecting the  steady-state equations for the core species
with the conservation laws of each of the systems. The number of points in this intersection might differ between extension models and the core model, depending on the form of the conservation laws.

We illustrate this using the two-site phosphorylation system  in Fig.~1A   and include dephosphorylation reactions,  
\begin{equation}\label{mod1}
S_2\rightarrow S_1,\qquad S_1\rightarrow S_0.
\end{equation}
 In addition, we add  the reactions,
 \begin{equation}\label{mod2}
0\rightarrow S_2,\qquad S_2\rightarrow 0.
\end{equation}
The motivation for the addition is not biological but for illustrative reasons. It  allows us to plot  the steady-state equations in two dimensions. We will consider the positive steady states of the core model in Fig.~1A together with \eqref{mod1} and \eqref{mod2}, and the extension model in Fig.~1C together with \eqref{mod1} and \eqref{mod2}, and $y_1=S_0+E, y_2=S_1+E,y_3=S_2+E$. The added reactions are core reactions and do not involve intermediates.  Since the substrate $S_2$ is  degraded ($S_2\to 0$), the total amount of substrate is no longer conserved and there is only one conservation law, namely that for the kinase (compare \eqref{consw2}). 

At steady state, the core and any extension model  fulfill the relation 
\begin{equation}\label{genericshape}
[S_0] = \frac{a_1}{[E]([E]+a_2)}
\end{equation}
for some constants $a_1,a_2>0$ that depend on the rate constants of each model (see appendix). 

The relation is obtained from the steady-state equations of the core model alone and  therefore must be fulfilled by all extension models  for matching rate constants (Theorem~3).
One can show that the concentrations of $[S_1]$ and $[S_2]$ at steady state are uniquely determined by $[E]$ and $[S_0]$ (see appendix).

For  a given conserved amount  for the kinase, 
the  steady-state concentrations are determined by the common points of  the graph of \eqref{genericshape} and the curve for the conservation law.  For the core model this curve is $E_{\tot}^\A=[E]$, which is a vertical line in the $([E],[S_0])$-plane. Since \eqref{genericshape} is strictly decreasing in $[E]$, it follows that there is a single steady state for any choice of $E_{\tot}^\A$ (Fig.\,2A). 

\begin{figure}[!t]
\begin{center}
\includegraphics{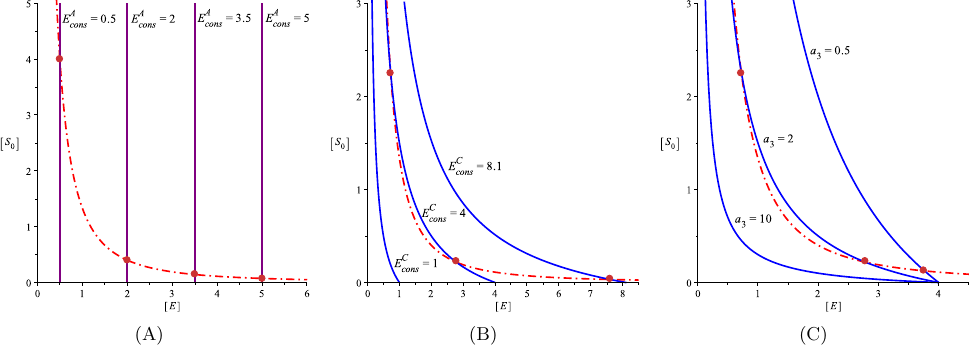}
\end{center}
\caption{\small The steady-state curve  \eqref{genericshape}  for $a_1=2,a_2=0.5$ (dashed-red) together with the curve for the conservation law. The steady states for a fixed conserved amount are  the  intersection points of the two curves (dashed and solid lines).
{\bf (A)} Core model. Conservation law curves (solid-purple)  for different  values of $E_{\tot}^\A$. {\bf (B)} Extension model. Conservation law curves (solid-blue) as  in \eqref{curveC} for different  values of $E_{\tot}^\C$  and $a_3=2$.  {\bf (C)}   Extension model. Conservation law curves (solid-blue) as in \eqref{curveC} for different  values of $a_3$  and $E_{\tot}^\C=4$.  }
\end{figure}

Consider  next the  extension model corresponding to Fig.\,1C  (with the modifications introduced in \eqref{mod1} and \eqref{mod2}).
For arbitrary fixed rate constants of the core model we choose rate constants of the extension model that realize the rate constants of the core model. This can always be achieved  for extension models with ``dead-end" complexes, like that of Fig.\,1C (see appendix). 
For $\dot{[Y]}=0$  the concentration of the intermediate is  $[Y]=a_3 [E][S_0]$ for some constant $a_3>0$ that depends on the rate constants of the extension model. Consequently,
\begin{equation}\label{curveC}
E_{\tot}^{\C} = [E] + a_3 [E][S_0],\quad \textrm{ or } \quad [S_0] = \frac{E_{\tot}^{\C} - [E]}{a_3[E]}  \qquad \textrm{provided that }a_3\neq 0. 
\end{equation}
If $a_3=0$ then we obtain the core model. In the particular case,  $a_3$ varies independently of $a_1,a_2$ and all values of $a_3$ can be obtained when realizing the rate constants of the core model.
Combining \eqref{genericshape} and \eqref{curveC} yields a second order polynomial in $[E]$: 
\begin{equation}\label{secorder}
a_1a_3=(E_{\tot}^C-[E])([E]+a_2).
\end{equation}
Hence, for fixed $a_1,a_2,a_3$, the polynomial  can have zero, one or two positive solutions, depending on  the value of $E_{\tot}^C$. 
Fig.\,2 shows graphically the steady-state solutions for the core (Fig.\,2A) and the extension  (Figs.\,2B-C) model as the intersection of the steady-state equation \eqref{genericshape} and the curve for the conservation law for different values of $E^A_{\tot}, E^C_{\tot}$ and $a_3$.  In Fig.\,2 the curve for the steady-state equation (dashed-red) is given for  $a_1=2$ and $a_2=0.5$ and is the same for the two models. 
For the core model, the conservation law curve is a vertical line (purple), which intersects the steady-state curve in precisely one point (Fig.\,2A). For the extended model, the conservation law curve (blue) is the ratio in \eqref{curveC}. Depending on the value of $E_{\tot}^C$, the two curves intersect in zero, one or two points illustrating how the number of steady states vary with $E_{\tot}^C$ (Fig.\,2B, with $a_3=2$). The same conclusion is obtained by varying  $a_3$ while keeping  $E_{\tot}^C$ fixed (Fig.\,2C, with $E_{\tot}^C=4$).

 In this particular case, we could find explicit expressions for the steady-state concentrations in terms of the conserved amounts and the rate constants. This is not always the case.

\section{Number of steady states}
In the example in the previous section one can choose rate constants and conserved amounts such that  the  extension model does not have a positive steady state, even though the core model has a positive steady state for all choices of rate constants. 
However, it is easy to see that  $a_3$ can always be chosen so small  that there is at least one positive solution for fixed $a_1,a_2$ and $E_{\tot}^C$. If $a_3\approx 0$ then the  contribution of $[E][S_0]$ in \eqref{curveC} becomes insignificant and the extension model is ``similar" to the core model. This is observed in Fig.\,2C: for small $a_3$, the curve for the conservation law is almost a vertical line. 

Therefore, in the example,  it is always possible to choose matching rate constants such the number of steady states in the extension model is at least as big as the number of steady states in the core model, for corresponding  conserved amounts.
This observation holds generally.  We now state the main result concerning the dynamical properties of extension models and the number of steady states:

\medskip
\noindent
{\bf Theorem 4: }
If the core model has $N$ non-degenerate\footnote[1]{A steady state is said to be non-degenerate if the Jacobian of the ODE system evaluated at the steady state is non-singular (see appendix).} positive steady states for some rate constants and conserved amounts, then any extension model that realizes the rate constants has at least $N$ corresponding non-degenerate positive steady states for some rate constants and conserved amounts. 
Oppositely, if the extension model has  at most one positive steady state for any rate constants and conserved amounts then the core model has at most one positive steady state for any matching rate constants and conserved amounts. 

The rate constants and conserved amounts can be chosen such that the correspondence preserves unstable steady states with at least one eigenvalue with non-zero real part and  asymptotical stability for hyperbolic steady states.  \hfill{$\qed$}

\medskip
The proof essentially relies on the observation in the previous example that a certain parameter ($a_3$ in the example) can be chosen so small that the extension model and the core model are almost identical at steady state. The relationship between a reaction network and a subnetwork has been studied previously, but in different contexts. For example in \cite{craciun-feinberg,joshi-shiu-II}, where  subnetworks are defined by (certain) subsets of reactions, or in \cite{joshi-shiu-II}, where subnetworks are defined by removing species from reactions. Characterizations similar to Theorem~4 about the number of steady states hold in these situations. 

In Fig.\,2, the  steady state in the extension model closest to the steady state in the core model (for the same conserved amount) inherits the stability properties of the steady state of the core model. In this case it is asymptotically stable. However, we cannot conclude anything about the other steady state in the extension model from the core model alone.

\section{Steady-state classes and canonical models} 

The observations made about the conservation laws and the steady-state equations (Theorems 1-3) 
suggest that it suffices to know what core complexes contribute to the conservation laws in order to compare the extension and core models at steady state. 
In Fig.\,1B, the core complexes $S_0+E,S_1+E$ contribute to the conservation laws for the kinase and substrate. Any other extension model, contributing the same core complexes to the conservation law, will result in equations for the steady states of the same form.
Specifically, if two extension models contribute the same core complexes to the conservation laws and  realize the same rate constants,\footnote[2]{Here it is also required that the constants $\mu_{Y,y}$ vary independently} then the two models are identical at steady state. In particular, we can apply Theorem 4 to any of the two models.

Therefore we can group extension models according to the core complexes that appear in the conservation laws. 
We say that two extension models belong to the same  {\it steady-state class} if they share the same core complexes in the conservation laws.
The complexes characterizing a steady-state class are called the {\it class complexes}. We can use  Theorem 2  to provide a graphical characterization of the classes:  the core complexes that contain a species  appearing in some conservation law are selected. If there exists a reaction from such a core complex to an intermediate, then the core complex is a class complex.  The class of the core model is the class with no class complexes.

In Fig.\,3, the graphical characterization is illustrated using the  core model  in Fig.\,1A, written in simplified form. All species appear in some conservation law and hence all core complexes can be class complexes.
Consider the extension models in Fig.\,1B and Figs.\,1C,D with  $y_1=S_0+E,y_2=S_1+E,y_3=S_2+E$. The extension model in Fig.\,1C belongs to the steady-state class with class complex $y_1$ because there is only one path from a core complex to an intermediate: $y_1\rightarrow Y$. Similarly, the extension models in Figs.~1B and 1D have class complexes $y_1,y_2$. 
 We conclude that  Figs.~1B and 1D are in the same class, while the models in Figs.~1A and 1C are in different classes  and have different equations.  In this case, Fig.\,1A  has always one steady state for any choice of conserved amounts  and rate constants, while Figs.~1B-1D can be multistationary (this is proven by direct computation of the steady states in the appendix).

\begin{figure}[!t]
\centering
\includegraphics{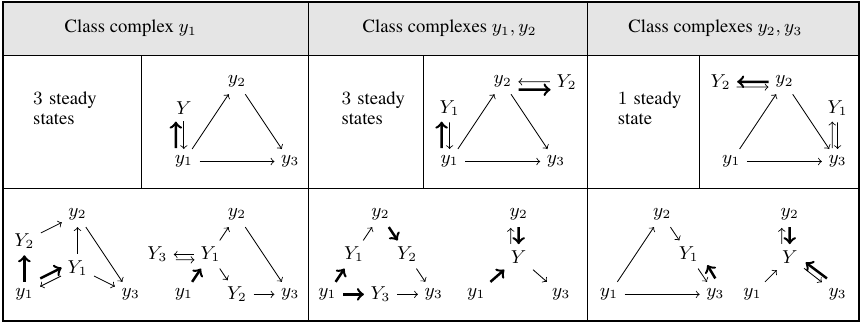}
\caption{\small We consider the core model of Fig.\,1A (with $y_1=S_0+E,y_2=S_1+E,y_3=S_2+E$) and its steady-state classes. Each class is characterized by an extension model (the canonical model) with a dead-end reaction added for each class complex (upper right corner). 
Each class (except the class of the core model) has an infinite number of members and a few of these are shown. Class complexes are source complexes of a reaction with an intermediate as product (marked in bold in the figure).
For the number of steady states we consider the model given in Fig.\,1A and  dephosphorylation reactions $S_2\rightarrow S_1$ and $S_1\rightarrow S_0$ (not shown in the figure). The number of steady states in each class refers to the maximal number of steady states that a model in the class can have for some choice of rate constants and total amounts.  This has been found by direct computation of the steady states (see appendix).  Alternatively, the CRN Toolbox could have been used \cite{crnttoolbox}.}
\end{figure}
 
Since class complexes characterize the steady-state classes, there is a finite number of  classes, at most  $2^K$,  with $K$ the number of core complexes ($K=3$ in Fig.\,3).  The classes are naturally ordered by set inclusion:  a class is smaller than another class if the latter contains the class complexes of the former. In particular, the steady-state class of the core model is smaller than any other  class. Thus, the class of Fig.\,1A is smaller than the classes of Fig.\,1B-1D, then classes of Fig.\,1B and Fig.\,1D are the same and the class of Fig.\,1C is smaller than the class of Fig.\,1B. The classes of the models in the first and the third box of  Fig.\,3 are not comparable as the first is $\{y_1\}$ and the last is $\{y_2,y_3\}$. 

All extension models in a steady-state class have common properties at steady state (subject to the requirement of realizability of rate constants). 
Thus, it is natural to select a  representative for each class with a small number of intermediates and such that the behaviors of all models in the class are reflected in the behavior of the representative.
To each class we construct a \emph{canonical model} by adding a dead-end reaction,  $y\rightleftharpoons Y$ for each class complex (see Fig.\,3 for an example).  Importantly, the steady-state equations for the canonical model are simpler than for any other extension model in the same class. It is shown in the appendix that the parameter space of the canonical model is a large as possible. 
This leads to the following corollary to Theorem~4.

\medskip
\noindent
{\bf Corollary 1.} If the canonical model of a steady-state class has a maximum of $N$ steady states  for any rate constants and conserved amounts, then all extension models in the class, or in any smaller  class, have at most $N$ steady states.  \hfill{$\qed$}

\medskip
In particular, if the largest canonical model (with a dead-end reaction added to  all core complexes)  is not  multistationary, then no extension model, including the core model, can  be multistationary. Likewise, if the smallest canonical model (the core model) is multistationary, then all extension models are multistationary.   If there are no conservation laws, then there is only one steady-state class and any steady state in the core model corresponds precisely to a steady state in the extension model (assuming rate constants are realizable; Theorems 2 and 3). Hence either all extensions models (with realizable rate constants) and the core model are multistationary or none of them are. Further, if the core model cannot have multistationarity neither can an extension model, independently of the realizability of the rate constants.
  
 Theorem~4 and Corollary~1 provide assistance to the model builder. First of all the modeler can focus on the canonical models only. 
 By screening the canonical models for the possibility of multistationarity,  the modeler  obtains a clear idea about the effects of  intermediates. 
In Fig.\,3, the steady-state class given by $\{y_2,y_3\}$ does not have multiple steady states, hence the same holds for the classes $\{y_2\}$, $\{y_3\}$ and the core model (Theorem 4). Multistationarity in Fig.\,3 (two first columns) is due to the non-linearity introduced by  $[y_1]$  in the conservation laws, irrespectively the presence or absence of  $[y_2]$ and $[y_3]$. 

\begin{figure}[!t]
\centering
\includegraphics{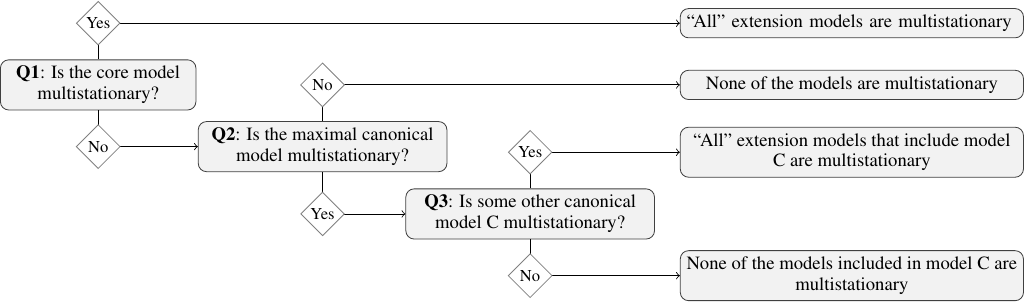}
\caption{\small
A decision tree to detect  multistationary steady-state classes.  ``All''  means that the model exhibits multistationarity as long as the rate constants of the core model can be realized by the extension model.
Q3 must be checked for different canonical models as necessary. }
\end{figure}

Our approach provides a simple graphical procedure to classify the extension models into a finite set of classes with common dynamical features, thereby elucidating the consequences of  choosing a specific model. Fig.\,4 shows a decision diagram  that guides the modeler 
through a number of possibilities. Each decision can be checked using various computational methods \cite{crnttoolbox,conradi-PNAS,Feliu-inj,PerezMillan} or by manually solving the system (a task that simplifies due to the simple form of the canonical models).

\section{Example: two-component systems}
Table\,1 shows a biological  application of the decision diagram in Fig.\,4. We consider three models of two-component systems of increasing complexity \cite{Igoshin:2008,salvado:plosone:2012}. The basic mechanism consists of a sensor kinase that autophosphorylates $\text{SK}\leftrightharpoons \text{SK}^*$ (here $^*$ indicates a phosphate group), the phosphate group is subsequently transferred to a response regulator $\text{RR}$ and  dephosphorylation of $\text{RR}^*$  is catalyzed by a  phosphatase $\text{Ph}$. This model is considered in Table\,1 (model A). Models B and C in Table\,1 consist of the first model enriched with more mechanisms. Model B, $\text{SK}$ has a bifunctional role and acts as a phosphatase, and likewise RR catalyzes dephosphorylation of SK. 
Model C is an enrichment of model B with dephosphorylation of $\text{SK}^*$ by a phosphatase $\text{T}$. 
Models B and C  in Table\,1 are core models of the models considered in  \cite{Igoshin:2008,salvado:plosone:2012}. Models B and C are not extension models of model A, nor of each other. All models considered in Table\,1 have the total amount of kinase and the total amount of response regulator conserved.

\begin{table}[!t]
\centering
\includegraphics{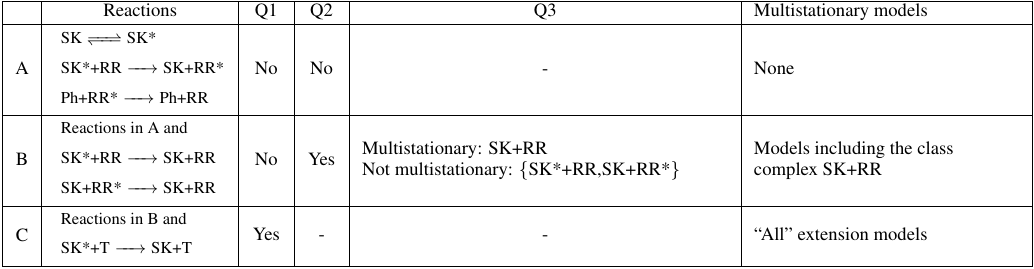}
\caption{\small
Example of an application of the decision tree in Fig.\,4. Four models of two-component systems are considered. All models are \emph{core} in the sense that they are not extension models of any  smaller models. SK=sensor kinase; RR=response regulator; Ph=phosphatase; T=phosphatase; *=phosphorylated (activated) state.  {\bf (A)} Basic phosphorelay mechanism: SK autophosphorylates and transfers the phosphate group to RR; a phosphatase dephosphorylates RR*.  {\bf (B)} Same as  {\bf (A)}, in addition SK is bifunctional and dephosphorylates RR* and RR catalyzes dephosphorylation of $\text{SK}^*$. 
 {\bf (C)} Same as  {\bf (B)} with the addition of a phosphatase T for SK*. System  {\bf (B)} is a core model of the mechanism considered in \cite{Igoshin:2008} and  in \cite[Model A]{salvado:plosone:2012}. {\bf (C)} is a core model of \cite[Model B]{salvado:plosone:2012}. The models analyzed in \cite{Igoshin:2008,salvado:plosone:2012} are extension models belonging to multistationary classes (last column of the table) and hence display multistationarity. The answers to Q1-Q3 have been obtained using the CRN Toolbox \cite{crnttoolbox}.
}
\end{table}

We have applied the decision tree in Fig.\,4 to each of the models. Model A and C are \emph{robust} with respect to the choice of intermediates: model A cannot exhibit multistationarity for any choice of rate constants and model C exhibits multistationarity for some choice of rate constants, independently of how intermediates are included in the models. Oppositely, model B is sensitive to how intermediates are introduced. The core model is not multistationary but inclusion of intermediates in some reaction paths introduces multistationarity. We conclude that modeling of this system needs to be done carefully, as the qualitative conclusions that can be drawn from the model depends on the choice of intermediates.  

Our analysis of the canonical models identify the steady-state classes that can exhibit multistationarity and pinpoint the particular class complexes that introduce non-linearity in the conservation laws. The analysis provides a simple overview of the effect of introducing intermediates in different reactions.

\section{Discussion}
Our work develops from the perspective of the model and clarifies the effects of intermediate species in biochemical modeling. Simplifications are always applied in  model building but generally on a case to case basis, motivated by biological assumptions. One example is the Quasi-Steady-State Approximation (QSSA), where equations of the form $[\dot{Y}]=0$, together with some (but not all) conservation laws,   are used to eliminate species \cite{segal,enz-kinetics}. This results in a \emph{hybrid} model   between our core and extension models. Our framework allows us to eliminate intermediate species generally and to compare core and extension models in a formal mathematical way. This comparison can be made independently of particular biological assumptions.
An important insight is that model simplification and model choice must be pursued with great care as crucial dynamical properties might change radically  by the inclusion of  intermediates. 

We remarked in the introduction that intermediates have been shown to affect steady-state properties of  a system, such as the emergence of ultrasensitivity \cite{Legewie:2005hw,Ventura-Hidden}. It follows from our results that intermediates cannot change a model's properties at steady state if there are no conservation laws. In particular, if production and degradation of each species are explicitly modeled, then a model without intermediates is fully justified at steady state. 

It has previously been noted that models that seem very similar  can have different qualitatively properties, e.g.\,\cite{craciun-feinberg-pnas}.  Our analysis is a  step forward in quantifying the relationship between simple and complex models of the same system, and in using simple models to predict properties of complex systems. Our results can guide the modeler through the critical issue of choosing  a model and in learning about model properties. 
As such the results are useful for interpretation of experimental  data and for designing synthetic  systems. 
We envisage that our techniques can  be extended to other models  than those defined by intermediates and can provide further insight into the nature of biochemical and other types of modeling \cite{rao,Fel_elim}.

\subsection*{Acknowledgements} 
 This work was supported by the Lundbeck Foundation, the Leverhulme Foundation and the Danish Research Council. E.F. is supported by a postdoctoral grant ``Beatriu de Pin\'os'' from the Generalitat de Catalunya and the project  MTM2012-38122-C03-01 from the Spanish ``Ministerio de Econom\'{\i}a y Competitividad''.  Part of this work was done while E.F. and C.W. visited Imperial College London in 2011. Neil Bristow is thanked for assistance.
The anonymous reviewers are thanked for their constructive comments.

\newpage

\appendix

\section{Proofs of theorems}

\paragraph{Erratum. }
The proof of Proposition 2 in the originally published version of the manuscript was erroneous. The result was though correct and the proof has been fixed in this version.

We are grateful to Magal\'{i} Giaroli from the University of Buenos Aires for pointing out the error in the proof of Proposition 2 in the previous version of the Electronic Supplementary Material. We would like to thank her and Daniele Cappelletti from University of Copenhagen for proof reading this new version.

\subsection{Preliminaries}

{\bf Reaction networks. } 
General standard background material on reaction networks can be found in \cite{gunawardena-notes,feinbergnotes}.
Here we recapitulate the definitions and properties necessary  for our work.
Consider a set $\mmS$ of $n$ species $S_1,\dots,S_n$.
A \emph{reaction network} (or simply network) consists of 
a set of reactions $\mmR$ whose elements take the form
$y\rightarrow y'$ with 
$y=\sum_{i=1}^n \alpha_{i} S_i$ and $y'= \sum_{i=1}^n \beta_{i} S_i $ for some non-negative integer coefficients
 $\alpha_{i},\beta_{i}\geq 0$. 
 The linear combinations $y,y'$ are called \emph{complexes} and the coefficients   are called \emph{stoichiometric coefficients}.
Complexes $y,y'$ can be seen as elements of the vector space $\R^n$ with entries given by the stoichiometric coefficients.
An \emph{intermediate}  $Y$  satisfies that 
the only complex involving $Y$ is $Y$ itself and there is at least one reaction of the form $y\rightarrow Y$ and one reaction of 
the form $Y\rightarrow y'$. Here $y$ and $y'$ can be other intermediates. An intermediate is thus both a species and a complex.

The molar concentration  of species $S_i$ at time $t$ is denoted by $c_i=c_i(t)$. To any complex $y$ we associate a monomial 
$c^y=\prod_{i=1}^n  c_i^{y_i}$. For example, if $y=(2,1,0,1)$, then the associated monomial is $c^y=c_1^2c_2c_4$. 
In the main text, concentrations are denoted by $[S_i]$ and the monomial associated to $y$ by $[y]$.

We assume that each reaction $y\rightarrow y'$ has an associated positive  \emph{rate constant} $k_{y\rightarrow y'}$, that is, $k_{y\rightarrow y'}$ is in $\R_+$. 
The set of reactions together with their associated rate constants give rise to a polynomial system of ordinary differential equations (ODEs) taken with 
\emph{mass-action kinetics}:
\begin{align}\label{ode}
\dot{c_i} &=\sum_{y\rightarrow y'\in \mmR} k_{y\rightarrow y'} c^y (y_i' -   y_i),\qquad i=1,\dots,n.
\end{align}
These ODEs describe the dynamics of the concentrations $c_i$ in time.
The steady states of the system are the solutions to a system of polynomial  equations in $c_1,\dots,c_n$ obtained by setting the 
derivatives of the  concentrations to zero:
\begin{align}\label{steadystate}
0 =&\sum_{y\rightarrow y'\in \mmR} k_{y\rightarrow y'} c^y (y_i' -   y_i), \qquad \textrm{for all }i=1\ldots,n.
\end{align}

It is convenient to treat  the rate constants as parameters with unspecified  values, that is as symbols. For that, let 
$$\Con=\{k_{y\rightarrow y'}| y\rightarrow y'\in \mmR \}$$
 be the set of the symbols. Then the system \eqref{steadystate}  is a system of polynomial  equations in $c_1,\dots,c_n$ with coefficients in the field $\R(\Con)$. 

The dynamics of a reaction network might preserve quantities that remain constant  over time. If this is the case, the dynamics  takes place in a proper invariant subspace of $\R^{n}$.  Let $x\cdot x'$ denote the Euclidian scalar product of two vectors $x,x'$ and $\omR^n_+$ the vectors with non-negative coordinates.

\begin{definition}\label{stoichiometricspace}
The \emph{stoichiometric subspace}  of a reaction network with reactions set $\mmR$ is the following subspace of $\R^n$: 
$$\Gamma = \langle y'-y|\, y\rightarrow y' \in \mmR\rangle.   $$
\end{definition}

By the definition of the mass-action ODEs, the vector $\dot{c}$ points along the stoichiometric 
subspace $\Gamma$. The \emph{stoichiometric class} of a concentration vector $c$ is $\{c+
\Gamma\}\cap\omR_{+}^n$.  Two steady states    $c,c'$ are called \emph{stoichiometrically compatible} if $c-c'\in 
\Gamma$. This is equivalent to $\omega\cdot c = \omega \cdot c'$ for all $\omega\in \Gamma^{\perp}$. 

In other words, if $\omega=(\lambda_1,\dots,\lambda_n)\in \Gamma^{\perp}$, then  $\sum_{i=1}^n \lambda_i \dot{c_i}=0$. This implies that the 
linear combination of  concentrations $\sum_{i=1}^n \lambda_i c_i$ is independent of time and thus determined by the 
initial concentrations of the system.  Such a relation is called a \emph{conservation law} and the value it takes in a stoichiometric class is called a \emph{conserved amount}.
In particular, any steady-state solution of the system preserves the conserved  amounts.
The vectors  $\omega\in \Gamma^{\perp}$, that is the conservation laws, are the vectors $\omega$ such that $\omega\cdot v=0$ for all $v\in \Gamma$.  If the generators of $\Gamma$ given in Definition \ref{stoichiometricspace} are written as the columns of a matrix $A$ (called the stoichiometric matrix), then the conservation laws are found as elements of the kernel  of the transpose of $A$.

\medskip
\textbf{Graphs. } Given a directed graph $G$ we call $\tau$ a \emph{spanning tree} of $G$ if $\tau$ is a directed subgraph of $G
$ with the same node set as $G$, and the undirected graph obtained by removing orientations from edges in $\tau$ is 
connected and acyclic.  A spanning tree $\tau$ is said to be \emph{rooted at v} if $v$ is a node in $\tau$, and the unique path 
from any other node $w \in \tau$ to $v$ is directed from $w$ to $v$.  $G$ is \emph{strongly connected} if for any (unordered) 
pair of nodes $v,w \in \tau$ there is a directed path from $v$ to $w$.  If $G$ is labeled then any spanning tree $\tau$ will inherit 
the labelling from $G$ in the obvious way.  For any labeled graph $G$ we define
\begin{align*}
\pi\left(G\right) = \prod_{x\xrightarrow{a} y \in G} a\ .
\end{align*}

\medskip
\textbf{Core and extended models. }
Consider a \emph{core model} with species $\mmS_C=\{S_1,\dots,S_n\}$, set of reactions $\mmR_C$ and let $\mmC_C$ denote the set of core complexes.
An \emph{extension model} (of the core model) has the following form: 
\begin{enumerate}[(i)]
\item The set of species is $\mmS_E = \mmS_C\cup \mmY$ with $\mmY$ a set of intermediates. Let $p$ be the cardinality of $
\mmY$.
\item The set of reactions $y\rightarrow y'$ obtained from collapsing the reaction paths in the extension model $y\rightarrow Y_1\rightarrow \dots \rightarrow Y_k \rightarrow y'$ with $Y_i\in \mmY$ and $y,y'\in \mmC_C$ equals $\mmR_C$.
\end{enumerate}

The set of reactions $\mmR_E$ is divided into four non-overlapping subsets: 
\begin{itemize}
\item[-] The reactions that are both in the extended and in the core model, $\mmR_{C\cap E}=\mmR_E\cap \mmR_C$.
\item[-] The reactions from a core complex to an intermediate, $\mmR_{C\rightarrow E}=\{y\rightarrow Y\in
\mmR_E\vert\  y\in\mmC, Y\in\mmY\}$.
\item[-] The reactions from an intermediate to a core complex, $\mmR_{E\rightarrow C}=\{Y\rightarrow y\in
\mmR_E\vert \ y\in\mmC, Y\in\mmY\}$.
\item[-] The reactions between two intermediates, $\mmR_{E\rightarrow E}=\{Y\rightarrow Y'\in\mmR_E\vert \ Y,Y'\in\mmY\}$.
\end{itemize}

We assume that the set of species of an extended model is ordered as $\{S_1,\dots,S_n,Y_1,\dots,Y_p\}$.
For simplicity, we let $c_i$ denote the concentration of $S_i$ for $i=1,\dots,n$ and $u_i$ the concentration of $Y_{i}$ for 
$i=1,\dots,p$.
The ODEs of the extended model consist of $n+p$ equations. Since intermediates do not interact with species $S_i$, the ODE 
equations do not have monomials involving both $c_*$ and $u_*$.

\subsection{Proof of Theorem 1}

Consider a  core model with species $\mmS_C=\{S_1,\dots,S_n\}$, set of reactions $\mmR_C$ and let $\mmC_C$ denote the set of core complexes. Let $\Gamma_C$ be the stoichiometric space.
Consider an extension model with  set of species  $\mmS_E = \mmS_C\cup \mmY$ with $\mmY=\{Y_1,\dots,Y_p\}$ a set of intermediates, and set of reactions   $\mmR_E$.
Let $\Gamma_E$ be the stoichiometric space of the extended model:
$$\Gamma_E = \langle y-y'\vert y\rightarrow y'\in\mmR_E\rangle.$$
For every reaction $y\rightarrow y'\in \mmR_C$, there exists a reaction path   $y \rightarrow Y_{i_1} \rightarrow  \ldots \rightarrow Y_{i_k} \rightarrow y'$, possibly with empty set of intermediates, such that each reaction belongs to $\mmR_E$. 
It follows that there is an inclusion 
\begin{equation}\label{inclusion}
\Gamma_C\hookrightarrow \Gamma_E
\end{equation}
 obtained by setting the coordinates $n+1,\dots,n+p$ to zero.

Let the reaction graph of a network be the graph with the complexes as nodes  and an (undirected) edge between any two complexes forming a reaction. Let the reaction graph of the core model have $J$ components. Then the reaction graph of the extension model also has $J$ components. Any reaction in the core model can be realized as a series of reactions in the extension model, by assumption. 
Hence the extension model cannot have more than $J$ components. We show that it has precisely $J$ components.  Consider intermediates $Y_{i_1},\ldots, Y_{i_k}$ such that $y - Y_{i_1} - \ldots - Y_{i_k} - y'$ is a series of reactions (here $-$ is either $
\rightarrow$ or $\leftarrow$) and $y,y'$ belong to different connected components of the core reaction graph. If the reactions are all in the same direction then either $y\rightarrow y'$ or $y'\rightarrow y$ is in the core model and hence $y,y'$ belong to the same connected component of the core reaction graph. If the reactions are in different directions,  let $Y_{i_j}$ be the first intermediate such that $\rightarrow Y_{i_j}\leftarrow$ or $\leftarrow Y_{i_j}\rightarrow$. By hypothesis, there exists a reaction path $Y_{i_j}\rightarrow \dots \rightarrow y''$ or $y''\rightarrow \dots \rightarrow Y_{i_j}$ respectively. Then, either $y\rightarrow y''$ and $y'\rightarrow y''$ or the reverse reactions are core reactions and hence $y,y'$ belong to the same connected component.

The statement of Theorem \ref{fact:cons} is:
\begin{theorem}\label{fact:cons}
The conservation laws in the core model are in one-to-one correspondence with the conservation laws in the extension model. The correspondence is obtained by adding the same linear combination of the  concentrations of the intermediates to the  conservation laws of the core model.
\end{theorem}

Theorem \ref{fact:cons} will follow from the lemmas below.

\begin{lemma}\label{fact31} Assume that the reaction graph of the core model has $J$ connected components (which we order) and for $j=1,\ldots,J$, select a complex $y^j$ in each component. Let $\omega=(\omega_1,\ldots,\omega_n)\in\Gamma_C^{\perp}$ and define $a_j=\omega\cdot y^j$. Define a vector $\widetilde{\omega}\in \R^{n+p}$ such that 
$$\widetilde{\omega}_i = \begin{cases}\omega_i & \textrm{for } i=1,\dots,n, \\   a_j, & \textrm{if }Y_{i-n}\textrm{ is in the $j$-th component and }i=n+1,\dots,n+p.\end{cases}$$
We have 
\begin{enumerate}[(i)]
\item $\widetilde{\omega}\in\Gamma_E^{\perp}$.
\item  If $\omega^1,\ldots,\omega^d$ form a basis of $\Gamma_C^{\perp}$ then $\widetilde{\omega}^1,\ldots,\widetilde{\omega}^d$ form a basis of $\Gamma_E^{\perp}$.  
\end{enumerate}
\end{lemma}
\begin{proof}
First of all, we check that $a_j$ is independent of the choice of $y^j$. Fix a component $C_j$ of the reaction graph of the core model. For any reaction $y\rightarrow y'$ in $C_j$, we have $\omega\cdot (y'-y)=0$ and hence   $\omega\cdot y'= \omega\cdot y$. Since $C_j$ is connected, $a_j$ is independent of the choice of $y^j$.
Note that if $y$ is a core complex, then $\omega\cdot y = \widetilde{\omega}\cdot y$.

To show (i), we need to show that $\widetilde{\omega}\cdot(y'-y)=0$ for all $y\rightarrow y'\in\mmR_E$. 
Since $\omega\in \Gamma_C^{\perp}$, the equality clearly holds if $y\rightarrow y'\in \mmR_{C\cap E}$. 
Consider $y\rightarrow Y_i\in \mmR_{C\rightarrow E}$. If $Y_i$ belongs to the $j$-th component, then we have
$\widetilde{\omega}\cdot Y_i = a_j = \omega\cdot y = \widetilde{\omega}\cdot y$. Therefore, $\widetilde{\omega}\cdot(Y_i-y)=0$. Similarly we check that $\widetilde{\omega}$ is orthogonal to all reactions in $\mmR_{E\rightarrow C}$ and $\mmR_{E\rightarrow E}$. This proves (i). 

To prove (ii) note that if $\omega^1,\ldots,\omega^d$ are linearly independent then so are $\widetilde{\omega}^1,\ldots,\widetilde{\omega}^d$.
Further, by the inclusion \eqref{inclusion}, $\dim(\Gamma_C)\leq \dim(\Gamma_E)$. Consequently,
$$ d = \dim(\Gamma_C^{\perp}) \geq \dim(\Gamma_E^{\perp}) \geq d,$$
from where it follows that $\dim(\Gamma_E^{\perp})= d$ and hence $\widetilde{\omega}^1,\ldots,\widetilde{\omega}^d$ is a basis of 
$\Gamma_E^{\perp}$.
 \end{proof}

\begin{lemma} 
For $\widetilde{\omega}=(\omega_1,\dots,\omega_{n+p})\in\Gamma_E^{\perp}$, define 
 $\widetilde{\omega}^{\pi} = (\omega_1,\dots,\omega_{n})$.
 We have
 \begin{enumerate}[(i)]
 \item   $\widetilde{\omega}^{\pi}\in\Gamma_C^{\perp}$.
 \item  If $\widetilde{\omega}^1,\ldots,\widetilde{\omega}^d$ form a basis of $\Gamma_E^{\perp}$ then $\widetilde{\omega}^{1\pi},\ldots,\widetilde{\omega}^{d\pi}$ form a basis of $\Gamma_C^{\perp}$. 
 \end{enumerate}
\end{lemma}
\begin{proof}
Any reaction  $y\rightarrow y'\in \mmR_C$ satisfies  $y'-y\in \Gamma_E$ under the inclusion \eqref{inclusion}. Hence $\widetilde{\omega}\cdot (y-y')=0$. Since any core complex $y$ has coordinates $n+1,\dots,n+p$  equal to zero, $\widetilde{\omega}\cdot y =\widetilde{\omega}^{\pi}\cdot y$. This proves statement (i).

To prove (ii) we use that   $\dim(\Gamma_C^{\perp})=\dim(\Gamma_E^{\perp})$ (see previous proof).
Let $\omega\in \Gamma_C^{\perp}$ and consider $\widetilde{\omega}\in \Gamma_E^{\perp}$ as defined in Lemma~\ref{fact31}. 
Since  $\widetilde{\omega}^1,\ldots,\widetilde{\omega}^d$ form a basis of $\Gamma_E^{\perp}$, we have
$$\widetilde{\omega} = \lambda_1 \widetilde{\omega}^1+ \ldots  + \lambda_d \widetilde{\omega}^d $$
for some $\lambda_i$. Since $\omega = \widetilde{\omega}^{\pi}$, by projecting onto the first $n$ coordinates we obtain
$$\omega = \lambda_1 \widetilde{\omega}^{1\pi}+ \ldots  + \lambda_d \widetilde{\omega}^{d\pi}. $$
Therefore, $\widetilde{\omega}^{1\pi},\ldots,\widetilde{\omega}^{d\pi}$ generate $\Gamma_C^{\perp}$ and hence they form a basis.
\end{proof}

Note that the constructions of the two lemmas above give the desired correspondence between conservation laws since
for all $\omega\in \Gamma_C^{\perp}$ we have 
$\omega =  \widetilde{\omega}^{\pi}$
and for all $\widetilde{\omega}\in \Gamma_E^{\perp}$ we have 
$\widetilde{\omega} = \widetilde{ (\widetilde{\omega})^{\pi}}$.

\begin{remark}
The results in this subsection show that core and extension models have the same deficiency \cite{feinbergnotes}. The deficiency of a network is defined as the number of complexes minus the dimension of the stoichiometric space minus the number of connected components of the reaction graph. We have proved that the core and any extension model have reaction graphs with the same number of connected components, and that both the dimension of the stoichiometric space and number of complexes of an extension model increase by the number of intermediates. As a consequence, the deficiency remains invariant.
\end{remark}

\subsection{Proof of Theorem \ref{fact:elim}}

The proof of Theorem~\ref{fact:elim} relies on ideas introduced in \cite{TG-rational} and developed generally in \cite{Fel_elim}.
Let us recall its statement with the notation introduced above:

\begin{theorem}\label{fact:elim}
The system of equations $\dot{u}_i=0$ for all intermediates $Y_i$ in the system can be solved in terms of the core species and  $u_i$ is expressed at steady state as a linear sum $u_i=\sum_y \mu_{i,y} c^y$.  A monomial $c^y$ appears in the expression if and only if there is a reaction path $y\rightarrow  \ldots\rightarrow Y_i$ involving  exclusively intermediates.
\end{theorem}
\begin{proof}
Let us consider the steady-state equations $\dot{u}_i=0$ for $i=1,\dots,p$ corresponding to the intermediates. These equations  take the form
\begin{equation}\label{eq:Y}
0 = \sum_{y\rightarrow Y_i\in \mmR_{C\rightarrow E}}  k_{y\rightarrow Y_i} c^y+  \sum_{Y_j\rightarrow Y_i\in \mmR_{E\rightarrow 
E}}  k_{Y_j\rightarrow Y_i} u_j   - \left(\sum_{Y_i\rightarrow y\in \mmR_{E\rightarrow C}}  k_{Y_i\rightarrow y}+\sum_{Y_i
\rightarrow Y_j\in \mmR_{E\rightarrow E}}  k_{Y_i\rightarrow Y_j} \right) u_i
\end{equation}
(here, $i$ is fixed and summation is over $Y_j$ and $y$).
It follows that equations \eqref{eq:Y} for $i=1,\dots,p$ form a system of linear equations in the variables $u_1,\dots,u_p$ and  coefficients in $\R[\Con \cup \{c_1,\dots,c_n\}]$. That is, equations \eqref{eq:Y} for $i=1,\dots,p$  form the linear system
\begin{equation}\label{auz}
A u+ z =0
\end{equation}
with  $u=(u_1,\dots,u_p)$,
and $A=\{a_{i,j}\}$, such that  for $i\neq j$ we have
$$ a_{i,j}= \begin{cases} k_{Y_j\rightarrow Y_i} & \textrm{if } Y_j\rightarrow Y_i\in \mmR_{E\rightarrow E} \\ 0  & 
\textrm{otherwise,}\end{cases} $$
and for $i=j$ we have
$$a_{i,i} = -e_i - d_i,\quad \textrm{with }\qquad e_i=\sum_{Y_i\rightarrow Y_k\in \mmR_{E\rightarrow E}}  k_{Y_i\rightarrow Y_k} ,
\quad 
d_i= \sum_{Y_i\rightarrow y\in \mmR_{E\rightarrow C}}  k_{Y_i\rightarrow y}.$$
We define $z=(z_1,\dots,z_p)$ to be the independent term:
$$z_i =  \sum_{ y\rightarrow Y_i\in \mmR_{C\rightarrow E}}  k_{y\rightarrow Y_i} c^y.$$

All coefficients but $a_{i,i}$ are positive. Further, $a_{i,j}\in \R[\Con]$ while $z_i\in \R[\Con \cup \{c_1,\dots,c_n\}]$.
The column sums of $A$ are not all zero. Indeed, the sum of the entries in column $i$ is
$\sum_{j=1}^p  a_{j,i}  = \sum_{j\!\colon\!\! j\neq i} a_{j,i} - e_i - d_i. $
Note that for $i$ fixed,
$$\sum_{j\!\colon\!\! j\neq i} a_{j,i} = \sum_{j\!\colon\!\! j\neq i} k_{Y_i\rightarrow Y_j} = e_i. $$
Therefore, we have that 
\begin{equation}\label{di}
\sum_{j=1}^p  a_{j,i}  = -d_i. 
\end{equation} Since by assumption $\mmR_{E\rightarrow C}$ is not empty, $d_i\neq 0$ for some $i$ and thus the column sums of $A$ are not all zero.

Consider  the labeled directed graph $\widehat{G}_{\mmY}$ with node set  $\mmY \cup \{*\}$.  We order the nodes such that  
$Y_i$ is  the $i$-th node and  $*$ the $(p+1)$-th  node.  The graph $\widehat{G}_{\mmY}$ has the following labeled directed 
edges:
\begin{itemize} \label{widehatG}
\item $Y_j\xrightarrow{a_{i,j}} Y_i$ if $a_{i,j}\neq 0$ and $i\neq j$, 
\item $Y_i\xrightarrow{d_i}  *$ if $d_i\neq 0$, and 
\item   $*\xrightarrow{z_i} Y_i$ if $z_i\neq 0$. 
\end{itemize}
All labels are in $\R[\Con \cup \{c_1,\dots,c_n\}]$ and are either zero or polynomials in $\Con \cup \{c_1,\dots,c_n\}$ with positive  coefficients. By definition of intermediates, the graph  $\widehat{G}_{\mmY}$  is strongly connected. Indeed, for every  intermediate  $Y_i\in \mmY$ there is a reaction path $Y_i\rightarrow Y_{j_1}\rightarrow \dots \rightarrow Y_{j_l}\rightarrow y'$ with $y'\notin \mmY$ and a reaction path $y\rightarrow Y_{j_1}\rightarrow \dots \rightarrow Y_{j_l}\rightarrow Y_i$ for some  $y\notin \mmY$. Therefore, there is a directed path in both directions between each intermediate and $*$ in $\widehat{G}_{\mmY}$, hence also between any two intermediates.

Let $\mathcal{L}=\{\lambda_{i,j}\}$ be minus the Laplacian matrix of $\widehat{G}_{\mmY}$.  If $i,j\leq p$, then $\lambda_{i,j}=a_{i,j}$. The entries of the last row of  $\mathcal{L}$ are $\lambda_{p+1,i} = d_i$ for $i\leq p$ and the entries of the last  column are $\lambda_{i,p+1}= z_i$ for $i\leq p$. By the Matrix-Tree theorem \cite{Tutte-matrixtree} we conclude that
$$(-1)^{p+i+j} \mathcal{L}_{(i,j)} =   \sum_{\tau \in \Theta(Y_j)}  \pi(\tau),  $$
in particular, since the $(p+1,p+1)$ principal minor of $\mathcal{L}$ is exactly $A$, we have
\begin{equation}\label{Spos}
\sigma:=  (-1)^p \det(A) = (-1)^p \mathcal{L}_{(p+1,p+1)} =   \sum_{\tau \in \Theta(*)}  \pi(\tau). 
\end{equation}
Since no spanning tree rooted at $*$ can involve a label $z_i$, $\sigma$ is in fact a polynomial in $\R[\Con]$. 
Since $\widehat{G}_{\mmY}$ is strongly connected, then there exists at least one spanning tree rooted at $*$, and hence $(-1)^{p}\det(A)$ is non-zero in $\R[\Con\cup \{c_1,\dots,c_n\}]$. It follows that  the system $Au+z=0$ has a unique solution in $\R(\Con \cup \{c_1,\dots,c_n\})$. 

 For $i=1,\dots,p$, we let $\sigma_i$ be the  following polynomial in $c_1,\dots,c_n$,
$$\sigma_i=(-1)^{i+1} \mathcal{L}_{(p+1,i)} =  \sum_{\tau \in \Theta(Y_i)}  \pi(\tau),$$ 
which is either zero or  has positive coefficients in $\R[\Con \cup \{c_1,\dots,c_n\}]$. 
 By Cramer's rule, we have
$$ u_i = \varphi_i(c_1,\dots,c_n)=\frac{(-1)^{1+i} \mathcal{L}_{(p+1,i)}  }{(-1)^p \mathcal{L}_{(p+1,p+1)}} = \frac{\sigma_i}
{\sigma}, 
\qquad i=1,\dots,p.$$
Since $\widehat{G}_{\mmY}$ is strongly connected,  there exists at least one spanning tree rooted at $Y_i$, and $\sigma_i\neq 0$ as a polynomial in $\R[\Con \cup \{c_1,\dots,c_n\}]$. 
  
Since $\sigma$ is a polynomial in $\R[\Con]$, then $u_i=\sigma_i/\sigma$ can be seen as a polynomial in   $\R[c_1,\dots,c_n]$ with 
coefficients in $\R(\Con)$. Further, each term $\sigma_i$ can be written as:
$$\sigma_i =\sum_{k=1}^p \alpha_{k,i} z_k  = \sum_{k=1}^p \alpha_{k,i}  \sum_{y\rightarrow Y_k\in \mmR_{C\rightarrow E}}  k_{y\rightarrow Y_k} c^y,$$
with $\alpha_{k,i}\in \R[\Con]$. Specifically, $\alpha_{k,i}$ is a sum of terms obtained from the spanning trees rooted at $Y_i$  containing the edge $*\rightarrow Y_k$. Each spanning tree gives a term, namely the products of its labels, except the label $z_k$ for the edge  $*\rightarrow Y_k$. 
If we define 
$$\mu_{i,y} = \sum_{k=1}^p \frac{\alpha_{k,i}  k_{y\rightarrow Y_k}}{\sigma}$$
(with $ k_{y\rightarrow Y_k}=0$ if the reaction $y\rightarrow Y_k$ does not exist)
then
\begin{equation}\label{eq:Yelim}
u_i = \sum_{y\in \mmC_C} \mu_{i,y} c^{y}.
\end{equation}
This proves the first part of the statement.

To prove the second part, we show that the coefficient $\mu_{i,y}$ can be obtained from a graphical procedure. For a fixed core complex $y$, let  $\widehat{G}_{\mmY}^y$ be the labeled directed graph with node set  $\mmY \cup \{*\}$ and nodes ordered as 
above. The graph $\widehat{G}_{\mmY}^y$ has the following labeled directed edges:
\begin{itemize}
\item $Y_j\xrightarrow{a_{i,j}} Y_i$ if $a_{i,j}\neq 0$ and $i\neq j$, 
\item $Y_i\xrightarrow{d_i}  *$ if $d_i\neq 0$, and 
\item   $*\xrightarrow{ k_{y\rightarrow Y_i}} Y_i$  if $ k_{y\rightarrow Y_i}\neq 0$. 
\end{itemize}
That is, $\widehat{G}_{\mmY}^y$ and $\widehat{G}_{\mmY}$ have the same edges and differ only in the label of the edges  $*\rightarrow Y_i$, $i=1,\ldots,p$. 
Then 
\begin{equation}\label{muiy}
\mu_{i,y}=\frac{\sigma_{i,y}}{\sigma_y}:=\frac{ \sum_{\tau \in \Theta^y(Y_i)}  \pi(\tau)}{\sum_{\tau \in \Theta^y(*)}  \pi(\tau)}
\end{equation}
 where $\Theta^y(\cdot)$ refers to 
the spanning trees of $\widehat{G}_{\mmY}^y$ rooted at the argument.
We have that $\mu_{i,y}\neq 0$ if and only if there is a spanning tree rooted at $Y_i$ in  $\widehat{G}_{\mmY}^y$. Equivalently, if and only if there exists a reaction path from $y$ (that is, $*$) to $Y_i$.
\end{proof}

\subsection{Proof of Theorem \ref{fact:subst}}
Let us recall the statement of Theorem \ref{fact:subst}.

\begin{theorem}\label{fact:subst} After substituting the expressions $u_i=\sum_y \mu_{i,y}c^y$ into the ODEs for $\dot{c}_i$ of the extension model, a mass-action system for the core model is obtained with rate constants that are  derived  from the reaction paths connecting the complexes in the extension model.
\end{theorem}
\begin{proof}
The system of equations that describes the mass-action kinetics of the core model for some constants $t_{y\rightarrow y'}$ is:
\begin{align}\label{eq:newsys}
\dot{c}_i &=  \sum_{y \rightarrow y'\in \mmR_C} t_{y\rightarrow y'} c^{y} (y_{i}' -   y_{i}).
\end{align}
The ODE corresponding to $\dot{c}_i$, $i=1,\dots,n$, of the extension model taken with mass-action kinetics is
$$\dot{c}_i = \sum_{y\rightarrow y'\in \mmR_{C\cap E}} k_{y\rightarrow y'} c^y (y'_i-y_i)  +   \sum_{j=1}^p \sum_{Y_j\rightarrow y'\in \mmR_{E\rightarrow C}}  k_{Y_j\rightarrow y'} u_j y'_i - \sum_{j=1}^p  \sum_{y\rightarrow Y_j\in \mmR_{C\rightarrow E}} k_{y\rightarrow Y_j} c^y y_i  .$$
Using \eqref{eq:Yelim}, we obtain
\begin{align}\label{eq:newsys2}
\dot{c}_i   & = \sum_{y\rightarrow y'\in \mmR_{C\cap E}} k_{y\rightarrow y'} c^y (y'_i-y_i)  +   \sum_{j=1}^p   \sum_{Y_j\rightarrow y'\in \mmR_{E\rightarrow C}} k_{Y_j\rightarrow y'}  \sum_{y\in \mmC_C} \mu_{j,y} c^{y} y'_i - \sum_{j=1}^p  \sum_{y\rightarrow Y_j\in \mmR_{C\rightarrow E}}  k_{y\rightarrow Y_j} c^y y_i.
\end{align}
We want to see that this expression can be written in the form of \eqref{eq:newsys} for some choice of constants $t_{y\rightarrow y'}$  expressed in terms of $k_*$. 
Let
$$\widetilde{k}_{y\rightarrow y'} =\sum_{j=1}^p  k_{Y_j\rightarrow y'}\mu_{j,y},\qquad A_i =  \sum_{j=1}^p    \sum_{y,y'\in \mmC_c}   \widetilde{k}_{y\rightarrow y'}  c^{y} y_i,\qquad 
B_i=  \sum_{j=1}^p  \sum_{y\rightarrow Y_j\in \mmR_{C\rightarrow E}} k_{y\rightarrow Y_j} c^y y_i. 
$$
where $\widetilde{k}_{y\rightarrow y'}$ might be zero  if $k_{Y_j\rightarrow y'}=0$ or $\mu_{j,y}=0$. 
Then \eqref{eq:newsys2} can be written as:
\begin{align*}
\dot{c}_i   & = \sum_{y\rightarrow y'\in \mmR_{C\cap E}} k_{y\rightarrow y'} c^y (y'_i-y_i) +  \sum_{y,y'\in \mmC_c} \widetilde{k}_{y\rightarrow y'}c^{y} (y'_i-y_i)+A_i-B_i.
\end{align*}
 Assume that  for all fixed $i$ we have $A_i=B_i$ (proven below). Then \eqref{eq:newsys2} reduces to
\begin{align}\label{eq:newsys3}
\dot{c}_i   & = \sum_{y\rightarrow y'\in \mmR_{C\cap E}} k_{y\rightarrow y'} c^y (y'_i-y_i)  +   \sum_{y,y'\in \mmC_c} \widetilde{k}_{y\rightarrow y'}c^{y} (y'_i-y_i).
\end{align}
Let us see that $\widetilde{k}_{y\rightarrow y'} \neq 0$ if and only if there is a reaction path from $y$ to $y'$ involving exclusively intermediates. If $\mu_{j,y}\neq 0$ then there is a spanning tree in $\widehat{G}_{\mmY}^{y}$ rooted at $Y_j$. In particular, there is
a reaction path from $y$ to $Y_j$ involving intermediates. If further $k_{Y_j\rightarrow y'}\neq 0$ then there is a reaction $Y_j\rightarrow y'$ which all together give a reaction path $y$ to $y'$. By hypothesis, the reaction $y\rightarrow y'$ is in the core model.

Reciprocally any reaction $y\rightarrow y'$ in the core model appears in at least one reaction path $y\rightarrow Y_{i_1}\rightarrow \dots \rightarrow Y_{i_k} \rightarrow y'$, potentially without intermediates. If the reaction itself is not in the extended model, then  $k_{Y_{i_k}\rightarrow y'}\neq 0$ and there is a directed path from $*$ to $Y_k$ in the graph $\widehat{G}_{\mmY}^{y}$. Since $\widehat{G}^y_{\mmY}$ is strongly connected by hypothesis, any such path can be extended to a spanning tree of $\widehat{G}_{\mmY}^{y}$ rooted at $Y_k$. It follows that for all reactions $y\rightarrow y'\in \mmR_C\setminus \mmR_E$ there exists an index $k$ for which $\mu_{k,y}k_{Y_k\rightarrow y'}\neq 0$.

Consequently, \eqref{eq:newsys3} can be written as
\begin{align*}
\dot{c}_i   & = \sum_{y\rightarrow y'\in \mmR_{C}} (k_{y\rightarrow y'} + \widetilde{k}_{y\rightarrow y'})  c^y (y'_i-y_i)
\end{align*}
(with $k_{y\rightarrow y'}=0$ if the reaction $y\rightarrow y'$ is not in the extended model).
Therefore, by defining
\begin{equation}
\label{ty}
t_{y\rightarrow y'} := k_{y\rightarrow y'}  +\widetilde{k}_{y\rightarrow y'}=k_{y\rightarrow y'}  + \sum_{j=1}^p  k_{Y_j\rightarrow y'}\mu_{j,y}
\end{equation}
a mass-action system of the core model is obtained.

It remains to show that for  fixed $i$ we have $A_i=B_i$.  It is sufficient to show that for  fixed $y\in \mmC_C$ with $y_i\neq 0$, we have 
 $$ \sum_{j=1}^p   \sum_{Y_j\rightarrow y'\in \mmR_{E\rightarrow C}} k_{Y_j\rightarrow y'}  \mu_{j,y} 
  =   \sum_{j=1}^p   k_{y\rightarrow Y_j}  $$ 
 where in the right-hand side we allow $k_{y\rightarrow Y_j}=0$ if the reaction does not exist.
 Consider the graph  $\widehat{G}_{\mmY}^y$ defined above. Recall that 
 $d_j= \sum_{Y_j\rightarrow y\in \mmR_{E\rightarrow C}}  k_{Y_j\rightarrow y}$ and $\mu_{j,y} =\frac{\sigma_{i,y}}{\sigma_y}$.
Therefore, we have to show that  for a fixed $y\in \mmC_C$ with $y_i\neq 0$
we have
 \begin{equation}\label{newtrees}
 \sum_{j=1}^p d_j \sigma_{j,y}   =  \sum_{j=1}^p  k_{y\rightarrow Y_j}\sigma_y.  
 \end{equation}

Consider the set $G_1$ of all possible subgraphs of $\widehat{G}_{\mmY}^{y}$ which are the union of a spanning tree rooted at 
$*$ and an edge from $*$ to some $Y_j \in \mmY$, and the set $G_2$ of all possible subgraphs of $\widehat{G}_{\mmY}^{y}$ 
which are the union of a spanning tree rooted at some $Y_j \in \mmY$ and an edge from $Y_j$ to $*$.  Observe that we can 
rewrite \eqref{newtrees} as
\begin{align*}
\sum_{\tau \in G_1} \pi(\tau) = \sum_{\tau \in G_2} \pi(\tau), 
\end{align*}
and so showing that \eqref{newtrees} holds reduces to showing that $G_{1} = G_{2}$.

Let $\tau \in G_{1}$.   There is a single cycle in $\tau$, containing at least the nodes $\ast$ and some node $Y_{k}$ to which the 
unique outward edge from $\ast$ points.  Along this cycle there is a unique inward edge to $\ast$, with label $d_m 
\neq 0$ for some $m$.  Note that there is a directed path from every node in $\tau$ to $\ast$.  The directed path from a node $w
$ to $\ast$ either passes through the node $Y_{m}$, or it does not.  In the former case, the directed path from $w$ to $Y_{m}$ is 
preserved if we remove the edge from $Y_{m}$ to $\ast$.  In the latter case, the path from $w$ to $\ast$ is unaffected if we 
remove the edge from $Y_{m}$ to $\ast$, and we can extend this path to $Y_{k}$ (via the edge from $\ast$ to $Y_{k}$), and (if 
$Y_{k} \neq Y_{m}$) hence to $Y_{m}$ (via edges which comprise part of the cycle in $\tau$).  We also know that the edge from 
$Y_{m}$ to $\ast$ is part of the unique cycle which $\tau$ contains.  Thus removing this edge yields a spanning tree of the same 
node set, but rooted at $Y_{m}$.  Since we know that $d_m \neq 0$, we can add this edge back in to see that $\tau 
\in G_{2}$.  This shows $G_{1} \subseteq G_{2}$.

The proof that $G_{2} \subseteq G_{1}$ is analogous, with the roles of $\ast$ and $Y_{k}$ reversed.  
\end{proof}

\subsection{Proof of Theorem \ref{fact:multi}}

We use the notation introduced in the previous sections. 
Consider a  core model with species set $\mmS_C$ and set of reactions $\mmR_C$. Consider an extension model with  species set $\mmS_E = \mmS_C\cup \mmY$ with $\mmY$  the set of intermediates, and  reaction set   $\mmR_E$.
Rate constants $t_{y\rightarrow y'}$ of the core model are \emph{realizable} in the extension model if there exist rate constants $k_{y\rightarrow y'}$ in the extension model such that
\begin{equation}\label{realization}
t_{y\rightarrow y'} =k_{y\rightarrow y'}  + \sum_{j=1}^p  k_{Y_j\rightarrow y'}\mu_{j,y},\end{equation}
which is the relationship established between parameters in the core and extension model in equation \eqref{ty}.

A steady state is said to be non-degenerate if the Jacobian of the ODE system at the steady state is non-singular over the stoichiometric space.

Let us recall  Theorem \ref{fact:multi} and Corollary \ref{fact:cor}:

\begin{theorem}\label{fact:multi}
If the core model has $N$ non-degenerate positive steady states for some rate constants and conserved amounts, then any extension model that realizes the rate constants has at least $N$ corresponding non-degenerate positive steady states for some rate constants and conserved amounts. Oppositely, if the extension model has at most one positive steady state for any rate constants and conserved amounts then the core model has at most one positive steady state for any matching rate constants and conserved amounts.

The rate constants and conserved amounts can be chosen such that the correspondence preserves unstable steady states with at least one eigenvalue with non-zero real part and asymptotical stability for hyperbolic steady states.
\end{theorem}

\begin{corollary} \label{fact:cor}
 If the canonical model of a steady-state class has a maximum of N steady states for any rate constants and conserved amounts, then all extension models in the class, or in any smaller class, have at most N steady states.
\end{corollary}

The theorem follows from the series of propositions and lemmas below. The corollary is a simple consequence of the theorem.

\begin{proposition}\label{Nnon}
Consider a  core model with species set $\mmS_C$ and set of reactions $\mmR_C$. Consider an extension model with  species set $\mmS_E = \mmS_C\cup \mmY$ with $\mmY$  the set of intermediates, and  reaction set   $\mmR_E$.
Assume that:
\begin{enumerate}[(i)]
\item For some choice of rate constants $\tau=\{t_{y\rightarrow y'}\}$, $y\rightarrow y'\in\mmR_C$, the core model 
has $N\geq 1$  distinct non-degenerate positive steady states  in the same stoichiometric class.
\item There exist rate constants $\kappa=\{k_{y\rightarrow y'}\}$ for the extension model that realize $\tau$, that is, rate constants such that 
$$t_{y\rightarrow y'} =k_{y\rightarrow y'}  + \sum_{j=1}^p  k_{Y_j\rightarrow y'}\mu_{j,y}.$$
\end{enumerate}
 Then,   there exists a choice of rate constants  for the extension model that realize $\tau$ for which there are $N$  distinct non-degenerate positive steady states  in the same stoichiometric class.  
\end{proposition}

\begin{proof}
We will first rewrite the steady-state equations for the core model and for the extension model in a way suitable for our purpose. Secondly we  show that  if the core model has  $N$ non-degenerate positive steady states in the same stoichiometric class then so does the extension model.
Let $d=\dim(\Gamma^{\perp}_C)=\dim(\Gamma^{\perp}_E)$ (Theorem \ref{fact:cons}). We assume that the extension model has $p$ intermediates and that the species set $\mmS_E$ is ordered as $S_1,\dots,S_n,Y_1,\dots,Y_p$ where  $\mmS_C=\{S_1,\dots,S_n\}$ and $\mmY=\{Y_1,\ldots,Y_p\}$. We let $c_i$ denote the concentration of $S_i$ and $u_i$ the concentration of $Y_i$.

Consider the core model, a concentration vector  $c\in\R^n_+$ and rate constants $\tau=\{t_{y\rightarrow y'}\}$. The steady-state equations are given by
$$g_{\tau}(c):=\sum_{y\rightarrow y'\in\mmR_C} t_{y\rightarrow y'} (y'-y)c^y=0,$$
together with the equations for the conservation laws for a given set of conserved amounts $T_1,\ldots,T_d$. We follow \cite{feliu-inj} and choose a reduced basis for $\Gamma_C^{\perp}$, that is, a basis $\{\omega^1,\ldots,\omega^d\}$ with $\omega^i=(\lambda^i_1,\ldots,\lambda^i_n)$ such that $\lambda^i_i=1$ and $\lambda^i_j=0, j\not=i$, $j\ge d$. Such a basis always exists, potentially by reordering the set of species $\mmS_C$ \cite{feliu-inj}. The system of equations to be solved can then be rephrased as 
$$\widetilde{g}_{\tau}(c)=0,\quad \text{where} \quad \widetilde{g}_{\tau}(c)=(\omega^1\cdot c-T_1,\ldots,\omega^d\cdot c-T_d,g_{\tau,d+1}(c),\ldots,g_{\tau,n}(c))$$
(see \cite{feliu-inj}).
Thus, two vectors $c,c'\in\R^n_+$ are steady states of the core model, for the rate constants $\tau$, in the same stoichiometric class   if and only if $\widetilde{g}_{\tau}(c)=\widetilde{g}_{\tau}(c')=0$ for some choice of $T_1,\ldots,T_d$.

Similarly, consider the extension model, a concentration vector  $(c,u)\in\R^{n+p}_+$, and rate constants $\kappa=\{k_{y\rightarrow y'}\}$. The steady-state equations are given by
$$0=f_{\kappa}(c,u)=\sum_{y\rightarrow y'\in\mmR_E} k_{y\rightarrow y'} (y'-y)c_1^{y_1}\cdot\ldots\cdot c_n^{y_n}u_1^{y_{n+1}}\cdot\ldots\cdot u_p^{y_{n+p}},$$
 together with the equations for the conservation laws for a given set of conserved amounts $T_1,\ldots,T_d$. 
 The conservation laws are related  to the conservation laws of the core model by Lemma~\ref{fact31} and we use the notation introduced there. It follows that if $\{\omega^1,\ldots,\omega^d\}$  is a reduced basis for $\Gamma_C^{\perp}$ then $\{\widetilde{\omega}^1,\ldots,\widetilde{\omega}^d\}$ is a reduced basis for $\Gamma_E^{\perp}$, and that the system of equations to be solved can be stated as
 \begin{eqnarray}\label{fcu}
 \lefteqn{\widetilde{f}_{\kappa}(c,u)=0,\quad \text{where}}  \qquad\nonumber \\
&\widetilde{f}_{\kappa}(c,u)=(\widetilde{\omega}^1\cdot (c,u)-T_1,\ldots,\widetilde{\omega}^d\cdot (c,u)-T_d,f_{\kappa,d+1}(c,u),\ldots,f_{\kappa,n+p}(c,u)).
\end{eqnarray}
Since $d\leq n$, the last $p$ components of $\widetilde{f}_{\kappa}(c,u)$ are the steady-state equations corresponding to $\dot{u}=0$.  Note that 
$$\widetilde{\omega}^i\cdot (c,u)-T_i=\omega^i\cdot c+\sum_{j=1}^p \widetilde{w}_{n+j}^i u_j -T_i,$$
 $i=1,\ldots,d$, where $\widetilde{w}_{n+j}^i$ is the $(n+j)$-th coordinate of $\widetilde{\omega}^i$ as defined in Lemma~\ref{fact31}.   
  Two vectors $(c,u),(c',u')\in\R^{n+p}_+$ are steady states of the extension model in the same stoichiometric class for the rate constants $\kappa$ if and only if $\widetilde{f}_{\kappa}(c,u)=\widetilde{f}_{\kappa}(c',u')=0$ for some choice of $T_1,\ldots,T_d$.
 
We will reformulate the equation $\widetilde{f}_{\kappa}(c,u)=0$  to obtain a system of equations that is closely related to the equation $\widetilde{g}_{\tau}(c)=0$. First recall that at steady state $u_i=\sum_{y\in \mmC_C} \mu_{i,y} c^{y}$ (Theorem \ref{fact:elim}).  In equation \eqref{fcu} we  will replace the functions $\widetilde{f}_{\kappa,i}(c,u)$, $i>n$,  by the functions $\widehat{f}_{\kappa,i}(c,u) =u_i-\sum_{y\in \mmC_C} \mu_{i,y} c^{y}$, $i>n$, and further replace the variables $u_j$, $j>n$, by  $\sum_{y\in \mmC_C} \mu_{j,y} c^{y}$   in $\widetilde{f}_{\kappa,i}(c,u)$, for all $i\leq n$.

Formally, we proceed in the following way.
Let $I_r$ denote the identity matrix of order $r$.
Note that the function $\widetilde{f}_{\kappa}(c,u)$ is linear in $u$ and can be written in block form as
$$\widetilde{f}_{\kappa}(c,u)=\left(\begin{array}{c} M  \\  A  \end{array}\right)u+\left(\begin{array}{c} v \\ z  \end{array}\right),$$ where $M$ is an $n\times p$ matrix with entries in $\R[\Con]$, $v$ a vector of length $n$ with components in $\R[\Con,c,T_1,\dots,T_d]$ and $A,z$ are given in the proof of Theorem \ref{fact:elim}, that is,
from equation~\eqref{auz}, we have that 
$$(\widetilde{f}_{\kappa,n+1}(c,u),\ldots,\widetilde{f}_{\kappa,n+p}(c,u))=Au+z.$$
The  $p\times p$ matrix $A$  has entries in $\R[\Con]$ and is  invertible  in $\R(\Con)$. The vector $z$ has length $p$ and depends on $c$ and $\Con$. 
 Let $A^{-1}$ be the inverse of $A$ in $\R(\Con)$ . By Theorem \ref{fact:elim}, the solution to $Au+z=0$ is given by $u_i=-(A^{-1}z)_i=\sum_{y} \mu_{i,y} c^{y}$. 
 Let $B$ be the $(n+p)\times(n+p)$ matrix defined in block form by
 $$B=\left(\begin{array}{cc} I_{n} & -MA^{-1} \\ 0 & A^{-1}  \end{array}\right). $$
 This matrix is invertible in $\R(\Con)$.
Then, the function $\widehat{f}_{\kappa}(c,u)$ defined by 
 \begin{equation}\label{BD}
 \widehat{f}_{\kappa}(c,u):= B \widetilde{f}_{\kappa}(c,u)
 \end{equation}
  fulfills
 $$
\widehat{f}_{\kappa,i}(c,u)  =    \begin{cases} \widetilde{f}_{\kappa,i}\Big(c_1,\dots,c_n,\sum_{y\in \mmC_C} \mu_{1,y} c^{y},\dots,\sum_{y\in \mmC_C} \mu_{p,y} c^{y}\Big)     & i=1,\ldots,n, \\
 u_i-\sum_{y\in \mmC_C} \mu_{i,y} c^{y} & i=n+1,\ldots,n+p.
\end{cases} $$
Indeed,
$$B \widetilde{f}_{\kappa}(c,u) =  \left(\begin{array}{cc} I_{n} & -MA^{-1} \\ 0 & A^{-1}  \end{array}\right) \left(\begin{array}{c} Mu+v \\ Au+z  \end{array}\right)  
= \left(\begin{array}{c} v-MA^{-1}z \\ u+A^{-1}z  \end{array}\right)  $$
and the claim follows from the equality $-(A^{-1}z)_i=\sum_{y} \mu_{i,y} c^{y}$.

Note that $\widehat{f}_{\kappa,i}(c,u)$, $i\leq n$, does not depend on $u$. Further, solving $\widetilde{f}_{\kappa}(c,u)=0$ is equivalent to solving $\widehat{f}_{\kappa}(c,u)=0$. Equation~\eqref{BD} ensures that the determinant of the Jacobian of $\widehat{f}_{\kappa}$ evaluated at $(c,u)$ is  non-zero if and only if the determinant of the Jacobian of $\widetilde{f}_{\kappa}$ evaluated at $(c,u)$ is non-zero. Consequently to study non-degenerate steady states of the extension model we can study zeros of $\widehat{f}_{\kappa}(c,u)$ for which the Jacobian is non-singular. This is what we do next.

Assume that the core model has $N$ positive non-degenerate steady states, 
$c^i\in\R^n_+$, $i=1,\ldots,N$, in the same stoichiometric class for some rate constants $\tau=\{t_{y\rightarrow y'}\}$, $y\rightarrow y'\in\mmR_C$. 
Let $T_1,\dots,T_d$ be the conserved amounts defining the stoichiometric class for the reduced basis $\{\omega^1,\ldots,\omega^d\}$.

Let $\kappa=\{k_{y\rightarrow y'}\}$ be  rate constants for the extension model \eqref{ty} such that
$$t_{y\rightarrow y'} =k_{y\rightarrow y'}  + \sum_{j=1}^p  k_{Y_j\rightarrow y'}\mu_{j,y}$$
for all reactions $y\rightarrow y'$ in the core model $\mmR_C$  (which exist by assumption).
Then by construction and using Theorem \ref{fact:subst} we have
$$\widehat{f}_{\kappa,i}(c,u) =\begin{cases} \widetilde{g}_{\tau,i}(c)+\sum_{j=1}^p \widetilde{w}_{n+j}^i \sum_{y} \mu_{j,y} c^{y} & i=1,\dots,d \\ \widetilde{g}_{\tau,i}(c), &i=d+1,\dots,n. \end{cases}$$

Let $\theta\in\R_+$ be a positive constant. Define a new set of rate constants $\kappa^{\theta}=\{k^{\theta}_{y\rightarrow y'}\}$ by $k^{\theta}_{y\rightarrow y'}= k_{y\rightarrow y'}/\theta$ if $y\rightarrow y'\in \mmR_{E\rightarrow E}$ or $\mmR_{E\rightarrow C}$ and $k^{\theta}_{y\rightarrow y'}= k_{y\rightarrow y'}$ otherwise. Let $t^{\theta}_{y\rightarrow y'}$ and $\mu^{\theta}_{j,y}$ correspond to  $t_{y\rightarrow y'}$ and $\mu_{j,y}$, respectively, obtained with the rate constants $\kappa^{\theta}$ using \eqref{ty} and \eqref{muiy}.
Then
$$\mu^{\theta}_{j,y}=\theta\mu_{j,y},\quad \text{ and } \quad t^{\theta}_{y\rightarrow y'}=t_{y \rightarrow y'}.$$
The function $\widehat{f}_{\kappa}^{\theta}(c,u)$ for the rate constants  $\kappa^{\theta}$ takes the form
$$\widehat{f}_{\kappa,i}^{\theta}(c,u)  = \begin{cases} \widetilde{g}_{\tau,i}(c) +\theta\big(\sum_{j=1}^p \widetilde{w}_{n+j}^i \sum_{y} \mu_{j,y} c^{y} \big)& i=1,\ldots,d, \\
 \widetilde{g}_{\tau,i}(c) & i=d+1,\ldots,n, \\
 u_i-\theta\big( \sum_{y\in \mmC_C} \mu_{i,y} c^{y}\big) & i=n+1,\ldots,n+p.
\end{cases}$$
We observe that the Jacobian of $\widehat{f}_{\kappa}^{\theta}$ at $(c,u)$, $J_{(c,u)}(\widehat{f}_{\kappa}^{\theta})$,  takes the block form
$$J_{(c,u)}(\widehat{f}_{\kappa}^{\theta}) = \left(\begin{array}{cc}  J_{c}(\widetilde{g}_{\tau})+\theta(*) & 0  \\ -\theta(*) & I_p  \end{array}\right) $$
where ``$(*)$'' indicates some matrix that we are not concerned with knowing the exact form of.

By continuity, the function $\widehat{f}_{\kappa}^{\theta}(c,u)$  is well defined for  all $\theta\in\R$.  
That is, there is a well defined and differentiable function
\begin{eqnarray*}
 \R^{n} \times \R^{p} \times \R & \xrightarrow{F_{\kappa}} & \R^{n+p} \\
 (c,u,\theta) & \mapsto & F_{\kappa}(c,u,\theta):= \widehat{f}_{\kappa}^{\theta}(c,u).
\end{eqnarray*}
For $\theta=0$, the vectors  
$(c^i,0)$, $1\leq i\leq N$, are non-negative steady states in the stoichiometric class of the extension model defined by the conserved amounts $T_1,\ldots,T_d$.
That is $F_{\kappa}(c^i,0,0)=0$ for all $i$.
 The Jacobian of $\widehat{f}_{\kappa}^{0}(c,u)$ has  the matrix in block form 
$$J_{(c,u)}(\widehat{f}_{\kappa}^{0}) = \left(\begin{array}{cc}  J_{c}(\widetilde{g}_{\tau})  & 0  \\  0 & I_p  \end{array}\right) .$$
Since the Jacobian matrices of $\widetilde{g}_{\tau}$ evaluated at $c^i$, $1\leq i\leq N$, are by assumption non-singular,  
 the Jacobian matrices of $\widehat{f}_{\kappa}^{0}(c,u)$ evaluated at $(c^i,0)$  are non-singular.
Therefore, the  Implicit Function Theorem applied to $F_{\kappa}$ at the point $(c^i,0,0)$  guarantees that there exists an interval $I_i=(-\phi_i,\phi_i)$, $\phi_i>0$, and an open neighborhood $U_i$ of $(c^i,0)$ such that for all $\theta\in I_i$ there is a steady state $(c^i(\theta),u^i(\theta) )\in U_i$ in the stoichiometric class defined by $T_1,\dots,T_d$ and with $(c^i(0),u^i(0) )=(c^i,0)$. By making $\phi_i$ sufficiently small, the interval $I_i$ can be chosen such that $c^i(\theta)$ is  positive (i.e. $U_i\subseteq \R^n_+\times \R^p$) and the Jacobian of $\widehat{f}_{\kappa}^{\theta}(c,u)$ evaluated at $(c^i(\theta),u^i(\theta) )$ is non-singular for all $\theta\in I_i$. Restrict $I_i$ to the positive part, $I_i^{+}=[0,\phi_i)$. Since $c^i(\theta)$ is positive if follows from the definition of $\widehat{f}_{\kappa}^{\theta}(c,u)$ that $u^i(\theta)$ is positive for all $\theta\in I_i^{+}$. Hence $(c^i(\theta),u^i(\theta) )$ is a positive non-degenerate steady state in the stoichiometric class defined by the conserved amounts $T_1,\ldots,T_d$ for all $\theta\in I_i^{+}$. 
Since $c^i\not= c^j$, for all $i\not=j$, then  by choosing $\phi^i$ small enough we are guaranteed that $\cap_{i=1}^N U_i=\emptyset$.

With these data, let $\widehat{\phi}=\min(\phi_i\vert 1\leq i\leq N)$.
Then, for all $\theta \in (0,\widehat{\phi})$ the rate constants $\kappa^{\theta}=\{k^{\theta}_{y\rightarrow y'}\}$, $y\rightarrow y'\in\mmR_E$, fulfill that the extended model has $N$ positive distinct non-degenerate steady states $(c^i(\theta),u^i(\theta) )$, $1\leq i\leq N$, in the  stoichiometric class defined by $T_1,\ldots,T_d$. This concludes the proof.
\end{proof}

\begin{remark}
A steady state in the core model has always a corresponding steady state in the extension model for any choice of matching rate constants, $\kappa$ and $\tau$. It follows from the following: a steady state $c$ in the core model always defines a steady state concentration $u$ for the intermediates. By construction $(c,u)$ is a steady state. If there are $N$ steady states in the core model in some stoichiometric class for some rate constants then we are however not guaranteed that $N$ corresponding  steady states  in the extension model are in the same stoichiometric class. If the stoichiometric space of the core model has full dimension then the stoichiometric space of the extension model has full dimension (Theorem \ref{fact:cons}). Consequently the $N$ steady states are always in the same stoichiometric class.
\end{remark}

\begin{lemma}\label{Amatrix}
Let $A$ be the $p\times p$ matrix in equation~\eqref{auz}. Then all $p$ eigenvalues of $A$ have negative real part, that is, if $\lambda$ is an eigenvalue of $A$ then $\Re(\lambda)<0$.
\end{lemma}
\begin{proof} 
We will need the following fact $(*)$:  A Metzler matrix $M$ is a square matrix with all off-diagonal entries non-negative. If $M$ is a Metzler matrix then $\exp(M)$ is a matrix with non-negative entries. If $M$ is a Laplacian matrix, then $-M$ is a Metzler matrix,  $\exp(M)$ is a matrix with non-negative entries and all column sums equal to one. This result and the Perron-Frobenius theorem used later in the proof can be found in \cite{berman}.
The argument we give holds generally for Metzler matrices with non-positive column sums, but we have not been able to find a reference to it in the literature.

Equation~\eqref{Spos} shows that $(-1)^p\det(A)$ is  a non-zero polynomial in $\R[\Con]$ with positive coefficients. Hence zero cannot be an eigenvalue of $A$ for any choice of rate constants.  By definition, $A$ is a Metzler matrix, and thus  $B=\exp(A)$ has  non-negative entries.  
We extend $A$ to a $(p+1)\times(p+1)$  matrix 
$$\tilde{A}=\left( \begin{array}{cc}
A & 0_p \\
(d_i)_{i=1,\ldots,p} & 0 
\end{array} \right), $$ 
where $0_p$ is the $p$-dimensional column vector with entries 0 and $d_i$ are defined in the proof of Theorem~\ref{fact:elim}.
By \eqref{di}, $-\tilde{A}$ is a Laplacian, and hence $\tilde{B}=\exp(\tilde{A})$ has  non-negative  entries and all column sums are equal to one. The matrix $\tilde{B}$ takes the form, 
$$\tilde{B}=\left( \begin{array}{cc}
B & 0_p \\ 
\tilde{D} & 1
\end{array} \right), $$
where $\tilde{D}$ is a $1\times p$ matrix with non-negative entries. It follows that the column sums of $B$ are less than or equal to one.

An eigenvalue $\mu$ of $B$ is related to an eigenvalue $\lambda=\lambda_1+i\lambda_2$ of $A$ by $\mu=\exp(\lambda)$.
Assume first that $A$ is irreducible. Hence also $B$ is irreducible. It follows from the Perron-Frobenius theorem that all eigenvalues $\mu$ of $B$ fulfill $\vert \mu\vert\leq r\leq 1$ (the maximal column sum) for some real number $r$ and that $\mu=r$ is an eigenvalue. Since $\lambda=0$ is not an eigenvalue of $A$, then necessarily   $r<1$. Hence 
for all eigenvalues $\mu=\exp(\lambda)$ of $B$, we have $e^{\lambda_1}=\vert \mu\vert<1$ and hence  $\lambda_1=\Re(\lambda)<0$ for all eigenvalues $\lambda$ of $A$.

If $A$ is not irreducible then $A$ can be written in the following form, potentially after reordering the intermediates,
 $$A=\left(\begin{array}{cccc} \tilde{A}_1 & \ldots & \ldots & \ldots \\ 0 & \tilde{A}_2  & \ldots & \ldots \\  0 & 0 & \ldots & \ldots \\  0 & 0 & 0 & \tilde{A}_k \end{array}\right), $$
where $k> 1$ and $\tilde{A}_1,\ldots,\tilde{A}_k$ are irreducible square matrices. Each $\tilde{A}_j$ fulfills the same properties as $A$ above, that is, $\tilde{A}_j$ is a Metzler matrix with non-positive column sums and with at least one negative column sum. The latter follows from the following.  Let $\mmY_j$ denote the set of intermediates corresponding to the rows of $\tilde{A}_j$. 
 and let $j_1,\dots,j_t$ be the corresponding ordered row indices.  
Using \eqref{di} and the definitions above it,  and the block diagonal form of $A$, we have that the column sums of $\tilde{A}_j$ are given by
$$\sum_{l=j_1}^{j_t}  a_{l,i}  = \sum_{l=1}^p  a_{l,i}- \sum_{l\notin \{j_1,\dots,j_t\}}   a_{l,i} = -d_i - \sum_{l\leq j_1}   a_{l,i}. $$
By definition of intermediate, there exists at least an intermediate $Y_i$ in $\mmY_j$ and a reaction $Y_i\rightarrow X$ with $X$ an intermediate not in  $\mmY_j$ or  a core complex. As a consequence, there exists an index $i$ such that $d_i\neq 0$ or $a_{l,i}\neq 0$ for some $l<j_1$. Therefore the column sums of $\widetilde{A}_j$ are not all zero.
Since the  eigenvalues of $A$ agree with the eigenvalues of $\tilde{A}_j$, $j=1,\ldots,k$, the lemma follows from considering each irreducible matrix $\tilde{A}_j$ by itself.
\end{proof}

\begin{remark}\label{zero-eigen} 
It follows from [\cite{feliu-inj},Remark 7.8] that for a non-degenerate steady state, the eigenvalues of the corresponding Jacobian matrix can be ordered such that $\lambda_i=0$ for $i=1,\ldots,d$ and $\lambda_i\not=0$ for $i=d+1,\ldots,n$, where $n-d$ is the dimension of the stoichiometric space.
\end{remark}


\begin{proposition}\label{stability}
Assume as in Proposition~\ref{Nnon}. Let $\dot{c}=g_\tau(c)$ be the ODEs describing the core model and $(\dot{c},\dot{u})=f_\kappa(c,u)$ the ODEs describing the extension model, for any $\kappa$ that realizes $\tau$. Let $\lambda_j$,  $j=1,\ldots,n$, be the eigenvalues of the Jacobian of $g_\tau$ evaluated at  a non-degenerate positive steady state $c'$, ordered such that $\lambda_i=0$ for $i=1,\ldots,d$ and $\lambda_i\not=0$ for $i=d+1,\ldots,n$. Further, let $\alpha_{i}$, $i=1,\ldots,p$, be the eigenvalues of the matrix $A$ in equation~\eqref{auz}.

Then  $\kappa$ can be chosen such that the extension model has $N$ non-degenerate positive steady states in the same stoichiometric class, each corresponding to one of the $N$ steady states of the core model, and such that  the following holds. Let  $\nu_j$,  $j=1,\ldots,n+p$, be  the eigenvalues of the Jacobian of $f_\kappa$ evaluated at the steady state $(c^*,u^*)$ corresponding to $c'$.
Appropriately ordered the eigenvalues fulfil:
\begin{itemize}
\item[(i)] $\nu_i=0, \quad i=1,\ldots, d.$
\item[(ii)]  If  \,\,  $\Re(\lambda_i)\not=0$ \,\, then \,\,  $\emph{sign}(\Re(\nu_i))=\emph{sign}(\Re(\lambda_i))$,  \quad $i=d+1,\ldots,n.$
\item[(iii)] $\emph{sign}(\Re(\nu_{i}))=\emph{sign}(\Re(\alpha_{i-n}))<0$,  \quad $i=n+1,\ldots,n+p.$
\end{itemize}
Consequently:
\begin{itemize}
\item[(iv)] If a steady state in the core model is unstable  and $\Re(\lambda_i)>0$, for some $i=d+1,\ldots,n$, then the corresponding steady state in the extension model is unstable.
\item[(v)]  If a steady state in the core model is hyperbolic,  that is, $\Re(\lambda_i)\not=0$ for $i=d+1,\ldots,n$, then the corresponding steady state in the extension model is hyperbolic
\item[(vi)] If a hyperbolic steady state in the core model   is asymptotically stable then the corresponding steady state in the extension model is hyperbolic and asymptotically stable.
\end{itemize}
\end{proposition}

\begin{proof}
We will make use of Schur's formula for the determinant of a square matrix $M$ with block form
$$M=\begin{pmatrix} A & B \\ C & D \end{pmatrix}.$$
If $D$ is a square invertible matrix, then
$$\det(M)=\det(D)\det(A-BD^{-1}C),$$
and similarly, if $A$ is a square invertible matrix, then
$$\det(M)=\det(A)\det(D-CA^{-1}B).$$

We use the notation introduced in the proof of Proposition~\ref{Nnon} and proceed as in the proof of that proposition.  We will be interested in the eigenvalues of the function $f_{\kappa}(c,u)$ and will start by making some preparations for understanding these. 

Let $\kappa=\{k_{y\rightarrow y'}\}$ be  rate constants that realize $\tau$  (which exist by assumption). We consider these constants fixed.
The function $f_{\kappa}(c,u)$ is linear in $u$ and can be written in block form as
$$f_{\kappa}(c,u)=\left(\begin{array}{c} M'  \\  A  \end{array}\right)u+\left(\begin{array}{c} v' \\ z  \end{array}\right),$$ where $M'$ is a real $n\times p$ matrix, $v'$ a vector of length $n$ depending on $c$ only,  and $A,z$ are given as in the proof of Theorem  \ref{fact:elim}, equation~\eqref{auz}, for the given $\kappa$.
Further, the vector $z$ has length $p$ and depends on $c$ only, and the  $p\times p$ matrix $A$   is  invertible   with inverse $A^{-1}$.  Let $B'$ be the $(n+p)\times(n+p)$ matrix defined in block form by
 $$B'=\left(\begin{array}{cc} I_{n} & -M'A^{-1} \\ 0 & A^{-1}  \end{array}\right). $$
 This matrix is invertible  with inverse
 $$B'^{-1}=\left(\begin{array}{cc} I_{n} & M' \\ 0 & A  \end{array}\right). $$
It follows that the function $\bar{f}_{\kappa}(c,u)$ defined by 
 \begin{equation}\label{BD2}
 \bar{f}_{\kappa}(c,u):= B' f_{\kappa}(c,u)
 \end{equation}
  fulfils
 $$\bar{f}_{\kappa,i}(c,u)  =    \begin{cases} f_{\kappa,i}\Big(c_1,\dots,c_n,\sum_{y\in \mmC_C} \mu_{1,y} c^{y},\dots,\sum_{y\in \mmC_C} \mu_{p,y} c^{y}\Big)     & i=1,\ldots,n, \\
 u_i-\sum_{y\in \mmC_C} \mu_{i,y} c^{y} & i=n+1,\ldots,n+p.
\end{cases} $$
Then by construction and using Theorem \ref{fact:subst} we have
$$\bar{f}_{\kappa,i}(c,u) = g_{\tau,i}(c), \quad i=1,\dots,n. $$

Define the   rate constants $\kappa^{\theta}=\{k^{\theta}_{y\rightarrow y'}\}$ for  $\theta\in\R_+$, identically to how we did in the proof of Proposition~\ref{Nnon}. Then the  function $\bar{f}_{\kappa}^{\theta}(c,u)$ for the rate constants  $\kappa^{\theta}$ takes the form
$$\bar{f}_{\kappa,i}^{\theta}(c,u)  = \begin{cases}
 g_{\tau,i}(c) & i=1,\ldots,n, \\
 u_i-\theta\big( \sum_{y\in \mmC_C} \mu_{i,y} c^{y}\big) & i=n+1,\ldots,n+p.
\end{cases}$$
We observe that the Jacobian of $\bar{f}_{\kappa}^{\theta}$ at $(c,u)$ does not depend on $u$, as $\bar{f}_{\kappa,i}^{\theta}(c,u)$ is linear in $u$. Further,  it is a 
 block matrix with form
$$J_{c}(\bar{f}_{\kappa}^{\theta}) = \left(\begin{array}{cc}  J_{c}(g_{\tau}) & 0  \\ -\theta Z_c & I_p  \end{array}\right), $$
where $Z_c$ is a $p\times n$ matrix that depends on $c$  only.  Define the block matrix $B'_\theta$ through its inverse
$$B'^{-1}_\theta=\left(\begin{array}{cc} I_{n} & \frac{1}{\theta}M' \\ 0 & \frac{1}{\theta}A  \end{array}\right), $$
and note that  $B'_\theta, B'^{-1}_\theta$ correspond to the matrices $B',B'^{-1}$ for the rate constants $\kappa^{\theta}$.
It follows that the Jacobian $J_{c}(f_{\kappa}^{\theta})$ of $f_{\kappa}^{\theta}$ at $(c,u)$,  is
\begin{equation}\label{eq:Jctheta}
J_{c}(f_{\kappa}^{\theta}) = B'^{-1}_\theta\left(\begin{array}{cc}  J_{c}(g_{\tau}) & 0  \\ -\theta Z_c & I_p  \end{array}\right) = \left(\begin{array}{cc}  J_{c}(g_{\tau})-M' Z_c & \frac{1}{\theta}M'  \\ -A Z_c&\frac{1}{\theta} A  \end{array}\right),
\end{equation}
which does not depend on $u$.
Further, the characteristic polynomial $\chi_c^\theta(x)$ of $J_{c}(f_{\kappa}^{\theta})$   is
\begin{align*}
\chi_c^\theta(x) &=\det(J_c(f^\theta_\kappa)-xI_{n+p}) \\
 &=\det\begin{pmatrix} J_c(g_\tau) -M'Z_c -xI_n& \frac{1}{\theta}M' \\ -AZ_c & \frac{1}{\theta} A-xI_p\end{pmatrix} \\
 &=\frac{1}{\theta^p}\det\begin{pmatrix} J_c(g_\tau) -M'Z_c-xI_n & M' \\ - AZ_c & A-\theta xI_p\end{pmatrix}. 
 \end{align*}
As we are interested in the eigenvalues of $J_{c}(f_{\kappa}^{\theta})$, that is, the zeros of $\chi_c^\theta(x)$ for $\theta>0$, it suffices to consider $\theta^p\chi_c^\theta(x)$.

We now assume that the core model has $N\geq 1$ non-degenerate positive steady states in the same stoichiometric class for  $\tau$. Proposition~\ref{Nnon} guarantees that there exists $\phi\in\R_+$, such that for $\theta\in(0,\phi)$  the extension model with $\kappa^\theta$ has $N$ non-degenerate positive steady states in the same stoichiometric class.
 Let the steady states in the core model be $c^j(0)$, $j=1,\ldots,N$, with corresponding steady states in the extension model being $(c^j(\theta),u^j(\theta))$, $j=1,\ldots,N$. These vary continuously in $\theta$ such that $ (c^j(\theta),u^j(\theta))\to (c^j(0),0)$ for $\theta\to 0$ (by the construction in the proof of Proposition~\ref{Nnon}). Thus, by taking $\phi$ potentially smaller,  we might consider each $\theta\mapsto (c^j(\theta),u^j(\theta))$ as a continuous function from $[0,\phi]$ into $\R^{n+p}_+$.

Each steady state will be treated individually. Therefore, we fix one steady state and suppress the index $j$. We write $(c(\theta),u(\theta))$ and $(c(0),0)$ (or just $c(0)$) for the fixed steady states  in the extension  and the core model, respectively.

We next turn to the function $\theta^p\chi_c^\theta(x)$ evaluated at a steady state $(c(\theta),u(\theta))$. Specifically, we consider the function
\begin{equation}\label{eq:h}
g\colon [0,\phi]\times \CC\to \CC, \qquad g(\theta,x)=\det\begin{pmatrix} J_{c(\theta)}(g_\tau) -M'Z_{c(\theta)}-xI_n & M' \\ - AZ_{c(\theta)} & A-\theta xI_p\end{pmatrix},
\end{equation}
which is continuous in $(\theta,x)\in [0,\phi]\times \CC$.
Using Schur's formula we find
\begin{align*}
g(0,x) &=\det\begin{pmatrix} J_{c(0)}(g_\tau) -M'Z_{c(\theta)}-xI_n & M' \\ - AZ_{c(0)} & A\end{pmatrix} \\
&=\det(A)\det(J_{c(0)}(g_\tau)-M'Z_{c(0)}-xI_n+M'A^{-1}AZ_{c(0)}) \\
&=\det(A)\det(J_{c(0)}(g_\tau)-xI_n),
\end{align*}
such that the zeros of $g(0,x)$ precisely are the eigenvalues $\lambda_1,\ldots,\lambda_n$, repeated according to multiplicity, of $J_{c(0)}(g_\tau)$.

To prove the proposition we will  make use of Hurwitz's theorem:
 
 \vspace{.2cm}
 \noindent
 {\bf Theorem.} (Hurwitz's theorem)
{\it Let $f_k\colon V\to\CC$, $k\in\N$, be a sequence of holomorphic functions defined on a connected open set $V\subseteq \CC$. Assume $f_k$, $k\in\N$, converge uniformly on compact subsets of $V$ to a holomorphic function $f\colon V\to \CC$. If $f$ has a zero of order $m$ at $z_0\in V$ then for every small enough $\rho > 0$ and for sufficiently large $k\in\N$ (depending on $\rho$), $f_k$ has precisely $m$ zeros in the disk defined by $\lvert z-z_0\rvert < \rho$, including multiplicity. Furthermore, these zeros converge to $z_0$ as $k\to\infty$.}

  \vspace{.2cm}
 The functions $g(\theta,x)$ fulfil the requirements of the theorem, where $\theta$ plays the role of the index $k$.
 All   matrices in the definition of $g( \theta,x)$ are continuous matrix functions on $[0,\phi]\times \CC$. It is a consequence of the continuity of $\theta\mapsto (c(\theta),u(\theta))$ in $\theta\in[0,\phi]$.
 Further, since $\theta\in[0,\phi]$ is compact,  the coefficients of $g(\theta,x)$ as a polynomial in $x$ are bounded continuous functions.
 Let $V\subseteq \CC$ be an open  bounded and connected set containing all zeros of $g(0,x)$, that is, containing all eigenvalues $\lambda_i$, $i=1,\ldots,n$, of   $J_{c(0)}(g_\tau)$. We   argue that for any compact set $K\subseteq V$, $g(\theta,x)\to g(0,x)$, $x\in K$, converge uniformly as $\theta\to 0$. It follows from continuity and boundedness of the coefficients and that $g(\theta,x)$ is a polynomial in $x$.  Finally, a polynomial is a holomorphic function.
 
 We might now apply Hurwitz's theorem with $V\subseteq \CC$ as above to the holomorphic functions $f_k(x)=g(\theta_k,x)$ for any sequence $(\theta_k)_{k\in\N}$ with $\theta_k\to 0$ as $k\to \infty$. As the result will not depend on the particular choice of sequence, the subindex $k$ will be omitted. Using Hurwitz's theorem, it follows that for $\rho>0$, there exists  $\phi(\rho)<\phi$, such that the function $g(\theta,x)$, $\theta\in(0,\phi(\rho))$ has at least as many zeros as $g(0,x)$ (with multiplicity), and such that $\lvert\nu_i(\theta)-\lambda_i\rvert<\rho$, where $\nu_i(\theta)$, $i=1,\ldots,n$, are  roots of  $g(\theta,x)$. Note that these roots are  eigenvalues of $J_{(c(\theta),u(\theta))}(f^\theta_\kappa)$.
 
 In particular, by choosing $\rho$ small,  the sign of the real parts of $\nu_i(\theta)$ and $\lambda_i$ agree if the real part of $\lambda_i$ is non-zero, that is,  for all $i=1,\ldots,n$,
 \begin{equation}\label{eq:realpart}
\Re(\lambda_i)\not=0\quad \Rightarrow \quad\sign(\Re(\nu_i(\theta)))=\sign(\Re(\lambda_i)),
 \end{equation}
 by ordering the eigenvalues appropriately.
 
 From now on we redefine $\phi$  such that \eqref{eq:realpart} is the case for all $\theta\in(0,\phi]$ (by choosing $\phi$ sufficiently small). The number of eigenvalues for $J_{c(\theta)}(f^\theta_\kappa)$ is $n+p$ and we have just established a relationship between $n$ of these and  the $n$ eigenvalues of $J_{c(0)}(g_\tau)$. We will next study the remaining $p$ eigenvalues of $J_{c(\theta)}(f^\theta_\kappa)$.

We will show that the remaining eigenvalues of $J_{c(\theta)}(f^\theta_\kappa)$ are close to $\frac{\alpha_i}{\theta}$, $i=1,\ldots,p$, for small $\theta$, where $\alpha_1,\dots,\alpha_p$ are the eigenvalues of $A$. To formalise this claim we do the following.
Let $\Omega\subseteq \CC\setminus\{0\}$ be an open connected and bounded set containing the eigenvalues $\alpha_1,\dots,\alpha_p$, and let $0<\phi'<\phi$ be such that 
$$k(\theta,x):=\det\left(J_{c(\theta)}(g_\tau) -M'Z_{c(\theta)}-\frac{x}{\theta}I_n \right)\not=0$$
for all $\theta\in(0,\phi']$ and $x\in \Omega$. This is possible   because the entries of the matrices $J_{c(\theta)}(g_\tau)$ and $M'Z_{c(\theta)}$ are bounded on compact intervals of $\theta$, and 
the set  $\Omega$ is bounded and does not contain 0. Hence we might choose $\phi'$ such that $\lvert\frac{x}{\theta}\rvert$ is  large enough and $k(\theta,x)\neq 0$ for all  $\theta\in(0,\phi']$ and $x\in \Omega$.

Consider now $\theta^p\chi_{c(\theta)}^\theta(\frac{x}{\theta})$ for $\theta\in(0,\phi']$ and $x\in \Omega$.
We might apply the second variant of Schur's formula to obtain an alternative expression for the characteristic polynomial:
\begin{align}
\theta^p\chi_{c(\theta)}^\theta\Big(\frac{x}{\theta}\Big) &= k(\theta,x) \det\left( A- xI_p+ AZ_{c(\theta)}\left(J_{c(\theta)}(g_\tau) -M'Z_{c(\theta)}-\Big(\frac{x}{\theta}\Big)I_n\right)^{-1}M'\right) \nonumber \\
&= k(\theta,x) h(\theta,x), \label{eq:largeX}
\end{align}
where the function $h\colon (0,\phi']\times \Omega\to \CC$ is defined by the last equality.
 For all  $\theta\in(0,\phi']$ and $x\in \Omega$, $k(\theta,x)\neq 0$ and hence any root of $\chi_{c(\theta)}^\theta(\frac{x}{\theta})$  satisfies  $ h(\theta,x)=0$.

We write
\begin{align}\label{eq:h}
h(\theta,x)
 &=\det\left( A- xI_p + AZ_{c(\theta)}G(\theta,x)M'\right),
 \end{align}
where 
$$G(0,x)=0,\quad\text{and}\quad G(\theta,x)^{-1}=J_{c(\theta)}(g_\tau) -M'Z_{c(\theta)}-\frac{x}{\theta}I_n$$
for $(\theta,x)\in(0,\phi']\times \Omega$. The matrix $G(\theta,x)$ and its inverse exist on $(0,\phi']\times \Omega$ by construction.

We will first argue that the function $G(\theta,x)$ can be extended to a  continuous on $[0,\phi']\times \Omega$ and  takes the following form:
$$G(\theta,x)=\frac{\adj(G(\theta,x)^{-1})}{\det(G(\theta,x)^{-1})}=\theta \widetilde G(\theta,x),\quad \text{for}\quad (\theta,x)\in (0,\phi']\times \Omega,$$
where $\adj(D)$ is the adjugate matrix of a square matrix $D$, and $\widetilde G(\theta,x)$ is a matrix whose entries are rational functions in $x$. The first equality follows from Cramer's rule. For the second equality, note that the entries $G_{kk'}(\theta,x)$, $k,k'=1,\ldots,n$, of $G(\theta,x)$ take the form
$$G_{kk'}(\theta,x)=\frac{\sum_{i=0}^{n-1} \left(\frac{x}{\theta}\right)^{\!i} a_i(c(\theta))}{\sum_{i=0}^n \left(\frac{x}{\theta}\right)^{\!i} b_i(c(\theta))}.$$
By multiplication with $\theta^n$ in the numerator and denominator we obtain 
$$G_{kk'}(\theta,x)=\theta \,\frac{\sum_{i=0}^{n-1} x^i \theta^{n-1-i} a_i(c(\theta))}{\sum_{i=0}^{n} x^i \theta^{n-i} b_i(c(\theta))},$$
such that $G(\theta,x)=\theta \widetilde G(\theta,x)$ for some matrix 	$\widetilde G(\theta,x)$, as claimed.
 The leading term of the denominator is always  non-zero and independent of $\theta$: $x^n b_n(c(\theta))=(-x)^n$; hence the denominator does not vanish. The function is well defined for all $(\theta,x)\in (0,\phi']\times \Omega$ by construction.
Further, the coefficients $a_i(c(\theta)),b_i(c(\theta))$ are bounded and continuous in $\theta\in[0,\phi']$; hence it follows that $\widetilde G(\theta,x)$ is continuous on $[0,\phi']\times \Omega$ and that $\widetilde G(\theta,x)$ converges as $\theta\to 0$.  Consequently,  also $G(\theta,x)$ is continuous on $[0,\phi']\times \Omega$ and that $G(\theta,x)\to 0$ as $\theta\to 0$.

We now return to the function $h(\theta,x)$ in \eqref{eq:h}. Using $G(\theta,x)=\theta \widetilde G(\theta,x)$, we have
\begin{align*}
h(\theta,x)& =\det\left( A- xI_p + \theta AZ_{c(\theta)}\widetilde G(\theta,x)M'\right), 
\end{align*}
and in particular,
$$h(0,x)=\det\left( A- x I_p\right), \qquad\text{and}\qquad
h(0,\alpha_i) =\det\left( A- \alpha_i I_p\right)=0.$$

Next, we will apply Hurwitz' theorem to $h(\theta,x)$, in a way similar to what we did for $g(\theta,x)$.  The functions $h(\theta,x)$ are defined on $x\in \Omega$, an open connected  and bounded set. For any compact set $K\subseteq \Omega$, the functions $h(\theta,x)$ converge uniformly as $\theta\to 0$. It follows from continuity and boundedness of the coefficients $a_i(c(\theta)),b_i(c(\theta))$. 

The functions are also holomorphic on $\Omega$ as the denominator of $\widetilde G(\theta,x)$ never vanishes on $\Omega$ (by construction); hence the derivative with respect to $x$ exists on $\Omega$ which implies that the functions are  holomorphic on $\Omega\subseteq\CC$.
 Hurwitz's theorem guarantees  that for small $\theta$ the number of zeros (with multiplicity) of $h(\theta,x)$ is the same as the number of zeros (with multiplicity) of $h(0,x)$, which is $p$. 
That is, there exist $\beta_i(\theta)$, $i=1,\dots,p$ zeros of $h(\theta,x)$ such that the distance between $\beta_i(\theta)$ and $\alpha_i$ is as small as desired. Then by \eqref{eq:largeX}, $\nu_{n+i}(\theta):=\beta_i(\theta)/\theta$ are zeros of the characteristic polynomial  $\chi_{c(\theta)}^\theta(x)$, for $i=1,\dots,p$.
 By choosing $\theta$ potentially smaller,  say $\theta<\phi''<\phi'$, we are guaranteed that 
 \begin{equation}\label{eq:bla}
 \sign(\Re(\nu_{n+i}(\theta))=\sign(\Re(\alpha_i))<0,
 \end{equation}
since $\sign(\Re(\alpha_i))<0$ by Theorem \ref{Amatrix}. The eigenvalues $\nu_{n+i}(\theta)$, $i=1,\ldots,p$, can be all made  different  from the previously determined eigenvalues $\nu_{i}(\theta)$, $i=1,\ldots,n$, as $\lvert \nu_{n+i}(\theta)\rvert$ become arbitrary large for $\theta$ arbitrary small.

We are now ready to prove the statements (i)-(iii). Let $\nu_i$, $i=1,\ldots,n+p$, be the eigenvalues of the extension model for some $\theta<\phi''$. By assumption the steady state $c(0)$ of the core model is non-degenerate. Since the dimension of the stoichiometric subspace of the core model is $n-d$ and the steady state is non-degenerate, then precisely $d$ of the eigenvalues $\lambda_i$ of the core model are zero (Remark~\ref{zero-eigen}). The dimension of the stoichiometric subspace of the extension model is $n+p-d$ (Lemma~\ref{fact31})  and  $d$ of the eigenvalues $\nu_i$ are zero. 
Using that $\nu_{n+i}(\theta)\neq 0$ for $i=1,\dots,p$, $d$ of the eigenvalues $\nu_i(\theta)$, $i=1,\dots,n$ are zero, and hence correspond precisely to the $d$ zero eigenvalues of $J_{c(0)}(g_\tau)$. We assume that these are ordered such that 
the zero eigenvalues are the first $d$.
Together with \eqref{eq:realpart}, this proves (i) and (ii). Item (iii) follows from \eqref{eq:bla}. Since there is a finite number of eigenvalues for all $N$ steady states, $\phi''$ can be chosen such that the statement is true for all eigenvalues.

Assume now $\theta$ is chosen such that (i)-(iii) are true. Consider an unstable steady state for the core model with $\Re(\lambda_{i})>0$ for some $i$. Then also $\Re(\nu_{i})>0$  according to (ii). Positivity of the real part of an eigenvalue implies that the steady state us unstable \cite{perko}, hence   the steady state in the extension model is unstable. It proves (iv). A steady state is hyperbolic if all eigenvalues have non-zero real part \cite{perko}. Then (v) follows from (ii) and (iii). A hyperbolic steady state is asymptotically stable if and only if all eigenvalues have negative real parts \cite{perko}. It follows that if a steady state in the core model is asymptotically stable then $\Re(\lambda_{i})<0$ for all $i=d+1,\ldots,n$. According to (ii) we also have $\Re(\nu_{i})<0$ for all $i=d+1,\ldots,n$. Together with (iii) the steady state in the extension model is asymptotically stable. It proves (vi).
\end{proof}

\section{Realization of rate constants}
\label{matching}

\subsection{Canonical models}\label{canonical}
Consider a  core model with species set $\mmS_C$ and set of reactions $\mmR_C$. Consider a canonical extension model with dead-end at some core complex $y^*$. That is, the extension model has set of species $\mmS_E = \mmS_C\cup \{Y\}$  ($Y$ is an intermediate), and set of reactions   $\mmR_E=\mmR_C \cup \{y^*\rightarrow Y,\ Y\rightarrow y^* \}$.

Consider  some choice of rate constants $\tau=\{t_{y\rightarrow y'}\}$, $y\rightarrow y'\in\mmR_C$ in the core model and conserved amounts $T_1,\dots,T_d$ corresponding to some choice of basis of $\Gamma_C^{\perp}$.
We prove here that there exist rate constants $\kappa=\{k_{y\rightarrow y'}\}$ for the extended model realizing $\tau$, that is, such that equation \eqref{realization} holds
$$t_{y\rightarrow y'} =k_{y\rightarrow y'}  + \sum_{j=1}^p  k_{Y_j\rightarrow y'}\mu_{j,y}.$$
 
In this case, there is only one intermediate and we have
$$\mu_{Y,y^*} = \frac{k_{y^*\rightarrow Y}}{k_{Y\rightarrow y^*}}. $$
Hence, for all reactions $y\rightarrow y'$ in $\mmR_C$, we have
$$t_{y\rightarrow y'} = k_{y\rightarrow y'}$$
and realization parameters obviously exist.

Further, any conservation law in the extended model that is not a conservation law in the core model takes the form 
$$ \widetilde{\omega} = \omega + aY$$
for some constant $a$. Written as an equation in the core species, we have
 $$ \widetilde{\omega} = \omega + a   \frac{k_{y^*\rightarrow Y}}{k_{Y\rightarrow y^*}} c^{y^*}.$$
We note that by varying the two rate constants $k_{y^*\rightarrow Y},k_{Y\rightarrow y^*}>0$ the coefficient of $c^{y^*}$ takes any desired non-zero value (if $a\neq 0$).

\subsection{Non-realizable constants}\label{sec:non-realizable}  
Consider the core model with reactions:
$$\xymatrix@R=12pt@C=35pt{& y_3 \\ y_1  \ar[ur]^{t_1}   \ar[r]^{t_2} \ar[dr]_{t_3} & y_4  & y_2. \ar[ul]_{t_4} \ar[l]_{t_5} \ar[dl]^{t_6} \\ & y_5  }$$
Consider the following extension model:
$$\xymatrix@R=3pt@C=35pt{   &   & y_3 \\ y_1   \ar[dr]^{k_1}  \\ & Y  \ar[uur]^{k_3}   \ar[r]^{k_4} \ar[ddr]_{k_5} & y_4  \\  y_2 \ar[ur]_{k_2} \\ & & y_5  }$$

Then we claim that this extension model cannot realize all choices of rate constants of the core model. If all rate constants of the core model, $t_1,\dots,t_6$ were realizable, then we could find rate constants $k_1,\dots,k_5$ such that equation \eqref{realization} holds, that is
\begin{align}\label{realization:ex}
t_1= & \frac{k_1k_3}{k_3+k_4+k_5}, & t_2=& \frac{k_1k_4}{k_3+k_4+k_5}, & t_3=& \frac{k_1k_5}{k_3+k_4+k_5},\\ t_4=& \frac{k_2k_3}{k_3+k_4+k_5}, & t_5=&\frac{k_2k_4}{k_3+k_4+k_5}, & t_6=&\frac{k_2k_5}{k_3+k_4+k_5}. \nonumber
\end{align}
Choose for instance
\begin{equation}\label{ex:t} t_1 = 3,\quad  t_2= 4,\quad t_3=5,\quad t_4=6,\quad t_5=8,\quad t_6=15.\end{equation}
Using $t_1,t_4$ and \eqref{realization:ex} we see that $k_2=2k_1$. Using $t_3$ and $t_6$ we see that $k_2=3k_1$ and hence system \eqref{realization:ex}  has no positive solution.

This conclusion can also be derived by noting that the core model has six independent parameters, $t_1,\ldots,k_t$, whereas the extension model has only five, $k_1,\ldots,k_5$.

\subsection{Deciding on realizability of rate constants}
In some cases, manual inspection suffices to decide whether an extension model can realize all choices of rate constants for the core model. However, it would be desirable to have an automated procedure to decide this.

We give here a necessary criterion that makes use of computational algebra tools, namely, Gr\"obner bases.
The realizability problem can be stated as follows. Let $m_C,m_E$ be the number of reactions in the core and extension models respectively. Consider the map
\begin{eqnarray*}
\R^{m_E}_+ & \xrightarrow{T} & \R^{m_C}_+ \\
\{k_{y\rightarrow y'}\} & \mapsto &    \left\{k_{y\rightarrow y'}  + \sum_{j=1}^p  k_{Y_j\rightarrow y'}\mu_{j,y}\right\}
\end{eqnarray*}
(assuming that the reaction sets are ordered). Asking for all choices of rate constants in the core model to be realizable in the extension model is equivalent to requiring that $T$ is a surjective map over the positive orthant $\R^{m_E}_+$. A minimal criterion is that $m_E\geq m_C$ (see Section~\ref{sec:non-realizable} for an example).

Since $\mu_*$ are rational functions in the rate constants $k_*$, $T$ extends to a rational map over $\R^{m_E}$ and the Zariski closure of the image of $T$ is a real algebraic variety that is defined by some ideal $I$ of $\R[t_1,\dots,t_{m_C}]$ \cite{cox:little:shea}. That is, $\overline{\im(T)}=V(I)$.
If $I$ is not the zero ideal,  then $T$ is not surjective over the positive orthant. Thus, for $T$ to be surjective,  a necessary condition is that $I$ is the zero ideal.

The ideal $I$ can be obtained using the \emph{Implicitization procedure} as described in \cite[\S 3]{cox:little:shea}. Let $k_1,\dots,k_{m_E}$, $t_1,\dots,t_{m_C}$ be variables corresponding to the rate constants in the core and extension models, respectively. Let $T=(T_1,\ldots,T_{m_C})$ and 
write the components $T_i$ as a quotient of polynomials in $k_1,\dots,k_{m_E}$: $T_i = f_i/g_i$. Let $J$ be the ideal of $\R[z,k_1,\dots,k_{m_E},t_1,\dots,t_{m_C}]$ given by
$$J=\langle g_1t_1-f_1,\dots,g_{m_C}t_{m_C} - f_{m_C},1- g_1\cdot \ldots \cdot g_{m_C} z\rangle.$$
The last polynomial can be dropped if all $g_i=1$ and if some $g_i$ are repeated, we consider them only once. Then $I$ is the elimination ideal
$$ I = J \cap \R[t_1,\dots,t_{m_C}].$$
A set of generators of $I$ is given by the polynomials involving $t_1,\dots,t_{m_C}$ only in the Gr\"obner basis of $J$ with the lexicographical order on the order of variables $z>k_1>\dots>k_{m_E}>t_1>\dots > t_{m_C}$. Therefore, if such a Gr\"obner basis has no polynomial in $t_1,\dots,t_{m_C}$, then any choice of rate constants of the core model is realizable in the extension model.

\medskip
{\bf Example. } Consider the example given in subsection \ref{sec:non-realizable}. The map $T$ is 
\begin{eqnarray*}
\R^{5}_+ & \xrightarrow{T} & \R^{6}_+ \\
(k_1,\dots,k_5)  & \mapsto &  \left(\frac{k_1k_3}{k_3+k_4+k_5},\frac{k_1k_4}{k_3+k_4+k_5},\frac{k_1k_5}{k_3+k_4+k_5},\frac{k_2k_3}{k_3+k_4+k_5},\frac{k_2k_4}{k_3+k_4+k_5},\frac{k_2k_5}{k_3+k_4+k_5}\right).
\end{eqnarray*}
It is clear that this map cannot be surjective over the positive orthant, but in general this will not  be the case.
The ideal $J$ is in this case:
$$J = \langle t_1-k_1k_3,t_2-k_1k_4,t_3-k_1k_5,t_4-k_2k_3,t_5-k_2k_4,t_6-k_2k_5,1-(k_3+k_4+k_5) z \rangle.  $$
Using Maple, we compute  the Gr\"obner basis $G_J$ of $J$ with the lexicographic order that orders $k_*,z$ larger than $t_*$  and obtain:
\begin{align*}
G_J = & \{-t_{{5}}t_{{3}}+t_{{6}}t_{{2}},-t_{{4}}t_{{3}}+t_{{6}}t_{{1}},-t_{{4}}t_{{2}}+t_{{5}}t_{{1}}, 
-t_{{5}}k_{{5}}+t_{{6}}k_{{4}},-k_{{5}}t_{{2}}+t_{{3}}k_{{4}},-t_{{4}}k_{{5}}+t_{{6}}k_{{3}},\\ &  -t_{{4}}k_{{4}}+t_{{5}}k_{{3}},-k_{{5}}t_{{1}}+t_{{3}}k_{{3}},-k_{{4}}t_{{1}}+k_{{3}}t_{{2}},-t_{{6}}+k_{{2}}k_{{5}},-t_{{5}}+k_{{2}}k_{{4}},-t_{{4}}+k_{{2}}k_{{3}},\\ & -t_{{3}}k_{{2}}+t_{{6}}k_{{1}}, -t_{{2}}k_{{2}}+t_{{5}}k_{{1}},-t_{{1}}k_{{2}}+t_{{4}}k_{{1}},-t_{{3}}+k_{{1}}k_{{5}},-t_{{2}}+k_{{1}}k_{{4}},-t_{{1}}+k_{{1}}k_{{3}}, \\ & zt_6+zt_5+zt_4-k_2, zt_6+zt_5+zt_4-k_1,-1+zk_3+zk_4+zk_5\}.
\end{align*}
It follows that 
$$I =\langle -t_{{5}}t_{{3}}+t_{{6}}t_{{2}},-t_{{4}}t_{{3}}+t_{{6}}t_{{1}},-t_{{4}}t_{{2}}+t_{{5}}t_{{1}} \rangle \neq 0 $$
 and hence $T$ is not surjective and there exist non-realizable rate constants. Observe that the polynomials in $I$ do not vanish when evaluated in the rate constants in \eqref{ex:t}.

Consider again the  core model  in subsection \ref{sec:non-realizable} but now with the extension model given by
$$\xymatrix@R=1pt@C=35pt{ y_1   \ar[dr]^{k_1}  &  &  y_3 \\ & Y_1 \ar[ur]^{k_3} \ar[dr]_{k_4}  \\   y_2 \ar[ur]_{k_2} &  & y_4 }\qquad\qquad \xymatrix@R=1pt@C=35pt{ y_1   \ar[dr]^{k_5}  &  &  y_4 \\ & Y_2 \ar[ur]^{k_7} \ar[dr]_{k_8}  \\   y_2 \ar[ur]_{k_6} &  & y_5.}$$
The map $T$ is given by
\begin{eqnarray*}
\R^{8}_+ & \xrightarrow{T} & \R^{6}_+ \\
(k_1,\dots,k_8)  & \mapsto &  \left(\frac{k_1k_3}{k_3+k_4},\frac{k_1k_4}{k_3+k_4}+\frac{k_5k_7}{k_7+k_8},\frac{k_5k_8}{k_7+k_8},\frac{k_2k_3}{k_3+k_4},\frac{k_2k_4}{k_3+k_4}+\frac{k_6k_7}{k_7+k_8},\frac{k_6k_8}{k_7+k_8}\right).
\end{eqnarray*}
The ideal $J$ is in this case:
$$J = \langle t_1-k_1k_3,t_2-k_1k_4-k_5k_7,t_3-k_5k_8,t_4-k_2k_3,t_5-k_2k_4-k_6k_7,t_6-k_6k_8,(k_3+k_4)(k_7+k_8)z\rangle.  $$
We proceed as above and compute the Gr\"obner basis $G_J$ in Maple.
In this case, we obtain  that 
$$I=0,$$
indicating that all rate constants might be realizable in this extension model. In fact, in this case we find that
$$\frac{t_3}{t_6}=\frac{k_5}{k_6},\quad  \frac{t_1}{t_4}=\frac{k_1}{k_2},\quad t_2+t_5=(t_1+t_4)\frac{k_4}{k_3}+(t_3+t_6)\frac{k_7}{k_8},$$
which has a positive solution $(k_1,\ldots,k_8)\in\R^8_+$ for any positive choice of $(t_1,\ldots,t_6)\in\R^6_+$.


\section{Information on the figures}

\subsection{Figure 2. Computation of the steady state curves}
Consider the reaction network with reactions
$$S_0+E \xrightarrow{k_1} S_1+E,\qquad  S_1+E \xrightarrow{k_2} S_2+E, \qquad S_0+E \xrightarrow{k_3} S_2+E $$
and 
$$S_2\xrightarrow{k_4} S_1,\qquad S_1\xrightarrow{k_5} S_0,\qquad 0\xrightarrow{k_6} S_2,\qquad S_2\xrightarrow{k_7} 0.$$
The mass-action ODE system is:
\begin{align*}
\dot{[S_0]} & =  -k_1[S_0][E] -k_3[S_0][E]  + k_5[S_1],  \\
\dot{[S_1]} & =  -k_2[S_1][E] -k_5[S_1] + k_1[S_0][E]+ k_4[S_2],  \\
 \dot{[S_2]} & = -k_4[S_2] + k_2[S_1][E] + k_3[S_0][E]+k_6 - k_7[S_2],   \\ 
 \dot{[E]} & = 0.
\end{align*}
Note that 
$$ \dot{[S_0]} +\dot{[S_1]} + \dot{[S_2]} = k_6 - k_7[S_2].$$
Hence at steady state
$$[S_2] = \frac{k_6}{k_7} $$
and the steady state value is independent of the other concentrations at steady state.
Using $\dot{[S_0]} =0$ we obtain that 
$$[S_1] =  \frac{k_1+k_3}{k_5} [S_0][E].$$
Therefore, the steady state concentration of $[S_1]$ is determined by those of $[E]$ and $[S_0]$.

Finally, using $\dot{[S_1]} =0$ we obtain that 
\begin{align*}
0 & =  -k_2[S_1][E] -k_5[S_1] + k_1[S_0][E]+ k_4[S_2]  \\ & = - \frac{k_2(k_1+k_3)}{k_5} [S_0][E]^2 -k_3[S_0][E] + \frac{k_4k_6}{k_7} 
\end{align*}
and hence
$$[S_0] = \frac{ \frac{k_4k_6}{k_7} }{[E]\left(\frac{k_2(k_1+k_3)}{k_5}[E] +k_3\right)  } =  \frac{ \frac{k_4k_6k_5}{k_7k_2(k_1+k_3)} }{\left([E] +\frac{k_3k_5}{k_2(k_1+k_3)}\right)[E]}.  $$
We have obtained the expression in Equation (9) in the main text, with $a_1= \frac{k_4k_6k_5}{k_7k_2(k_1+k_3)} $ and $a_2=\frac{k_3k_5}{k_2(k_1+k_3)}$.

The extended model for Figure 2 consists of the addition of the intermediate $Y$ with the reactions
$$S_0 + E \xrightarrow{k_8} Y,\qquad  Y\xrightarrow{k_9}S_0+E.$$
We showed in subsection \ref{canonical} that at steady state
$$[Y] = k_8/k_9 [S_0][E].  $$
Therefore, with the notation in the main text, $a_3=k_8/k_9$.

\subsection{Figure 3. Computation of the steady states}
We consider a choice of complexes $y_1,y_2,y_3$ in Figure 3 that represents a two-site phosphorylation event and compute the maximal number of steady states that the core model and the canonical models in Figure 3 can have.
Through this section, we are only interested in positive steady states and hence, when saying steady state we implicitly mean positive steady state.

Specifically, we consider a substrate $S$ that has two phosphorylation sites, with phosphorylation and dephosphorylation being sequential. 
We let $S_0$ denote the unphosphorylated substrate, $S_1$ denote the substrate with the first site phosphorylated and $S_2$ 
denote the fully phosphorylated form.
Phosphorylation reactions are
\begin{equation}\label{reactions}
 S_0+E\xrightarrow{k_1} S_1+E\qquad S_1+E\xrightarrow{k_2} S_2+E\qquad S_0+E\xrightarrow{k_3} S_2+E
 \end{equation}
where $E$ is a kinase. By setting $y_1=S_0+E$, $y_2=S_1+E$ and $y_3=S_2+E$ this model is an instance of the core model in 
the main text, Figure 1A. This model however accumulates at steady state all the substrate concentration in $S_2$. In order to have a more interesting analysis, 
we add simple
dephosphorylation reactions
\begin{equation}\label{reactions2}
S_2\xrightarrow{k_4} S_1\xrightarrow{k_5} S_0.
\end{equation}

We want to determine how many steady states can the different canonical representatives of each class have.
If we were simply interested in determining if the system can have multiple steady states or not, then we could use one of the several available automatized methods (e.g. \cite{crnttoolbox,feliu-inj}).

\medskip

{\bf Core model. }
The ODEs of the core model (Figure 1A) with reactions in \eqref{reactions} and \eqref{reactions2} are the following:
\begin{align}
\dot{[S_0]} & =  -k_1[S_0][E] -k_3[S_0][E]  + k_5[S_1],  \label{eq:1} \\
\dot{[S_1]} & =  -k_2[S_1][E] -k_5[S_1] + k_1[S_0][E]+ k_4[S_2],  \\
 \dot{[S_2]} & = -k_4[S_2] + k_2[S_1][E] + k_3[S_0][E],  \label{eq:3}  \\ 
 \dot{[E]} & = 0.
\end{align}
This system has two conservation laws:
$$S_{tot} =  [S_0]+[S_1]+[S_2],\qquad E_{tot}=[E], $$
that is, the concentration of kinase is clearly constant.

The steady-state equations are obtained by setting the left-hand side of the ODEs to zero. 
Using the steady-state equation derived from \eqref{eq:1} we obtain that
$$[S_1] = \frac{(k_1+k_3)E_{tot}}{k_5} [S_0],$$
and then using this expression and \eqref{eq:3} we have
$$[S_2] = \left(\frac{k_2(k_1+k_3)E_{tot}^2}{k_4k_5} + \frac{k_3E_{tot}}{k_4}\right)  [S_0].$$
Using the conserved amount $S_{tot}$ we obtain that at steady state
$$ [S_0] =S_{tot} \left( 1+ \frac{k_1k_4+k_3k_4+k_3k_5}{k_4k_5} E_{tot} +\frac{k_2(k_1+k_3)}{k_4k_5}E_{tot}^2\right)^{-1}.$$ 
Given any positive rate constants $k_*$ and positive conserved amounts $E_{tot}$ and $S_{tot}$, $[S_0]$ is positive and uniquely determined at steady state by this expression. Further, the steady-state value of $[S_0]$ determines the steady-state values of $
[S_1]$ and $[S_2]$ using the expressions above.
We conclude that the core model has  one positive steady state for each choice of rate constants and conserved amounts.

\medskip
{\bf Extension model 1. } We consider the canonical representative of the class in the first column of Figure 3. The reactions of the model are those in \eqref{reactions} and \eqref{reactions2} together with
$$ S_0 +E \xrightarrow{k_6} Y,\qquad Y\xrightarrow{k_7} S_0+E.$$
The ODEs of the model are the following:
\begin{align}
\dot{[S_0]} & =  -k_1[S_0][E] -k_3[S_0][E] - k_6[S_0][E]  + k_5[S_1] + k_7[Y],  \label{eq:1a} \\
\dot{[S_1]} & =  -k_2[S_1][E] -k_5[S_1] + k_1[S_0][E]+ k_4[S_2],  \\
 \dot{[S_2]} & = -k_4[S_2] + k_2[S_1][E] + k_3[S_0][E],  \label{eq:3a}  \\ 
 \dot{[E]} & = -k_6[S_0][E] +k_7[Y],\\
\dot{[Y]} &=  -k_7[Y] +k_6[S_0][E] \label{eq:5a}.
\end{align}
This system has two conservation laws:
$$S_{tot} = [S_0]+[S_1]+[S_2]+[Y],\qquad E_{tot}=[E]+[Y]. $$
Observe that the conservation laws of this system and the conservation laws of the core model are in correspondence, as indicated by Theorem \ref{fact:cons}.
Isolating $[E]$ from the kinase conservation law we have
$$ [E] = E_{tot} -[Y],$$
which is positive provided $0< [Y] < E_{tot}$. 
From the steady-state equation derived from \eqref{eq:5a} we have that
$$[S_0]=\frac{k_7[Y]}{k_6(E_{tot} - [Y])}. $$
Using the steady-state equation corresponding to \eqref{eq:1a}+\eqref{eq:5a} and \eqref{eq:3a} we iteratively obtain
$$[S_1] = \frac{(k_1+k_3)k_7}{k_5k_6}[Y],\quad  [S_2] =  \frac{k_3k_7}{k_4k_6} [Y] +\frac{k_2(k_1+k_3)k_7}{k_4k_5k_6}[Y]
(E_{tot} - [Y]).$$

Given $0< [Y] < E_{tot}$, all the steady-state expressions above are positive. The value at steady state of $[Y]$ is found by 
imposing the substrate conservation law ($S_{tot}$) to be fulfilled:
\begin{align}\label{stot}
 S_{tot} & = [Y] +  \frac{k_7[Y]}{k_6(E_{tot} - [Y])}  + \frac{(k_1+k_3)k_7}{k_5k_6}[Y] +\frac{k_3k_7}{k_4k_6} [Y] +
\frac{k_2(k_1+k_3)k_7}{k_4k_5k_6}[Y](E_{tot} - [Y]). 
 \end{align}
Let us focus on the right-hand side of this expression such that $S_{tot}=\varphi([Y])$.
The function $\varphi$ is continuous for $0< [Y] < E_{tot}$ and tends to infinity as $[Y]$ approaches $E_{tot}$. When $[Y]=0$ we 
further have $\varphi(0)=0$. It follows that for any given $S_{tot}$ there exists $[Y]\in [0,E_{tot})$ such that $S_{tot}=\varphi([Y])$ 
and hence a positive steady state exists.

If the function $\varphi$ is always increasing, then there is exactly one. If it can decrease in some part, then there can be more than one. 
Note that equation \eqref{stot}  can be rewritten as a polynomial of degree $3$ in $[Y]$ such that 
the roots in  $0< [Y] < E_{tot}$ are the positive steady states.
  It follows that there can be at most three positive steady states.

Each summand in \eqref{stot} is an increasing function of $[Y]$, except for the last summand. Hence, it is not clear whether $\varphi$ can be decreasing in some interval. The derivative of $\varphi$ with respect to $[Y]$ is 
$$\varphi'([Y]) =1+\frac{k_7}{k_6}\left( \frac{E_{tot}}{(E_{tot} - [Y])^2} + \frac{k_1+k_3}{k_5}+\frac{k_3}{k_4} +\frac{k_2(k_1+k_3)}
{k_4k_5}(E_{tot} - 2[Y])\right). $$
This derivative is negative if and only if
$$1+\frac{k_7}{k_6}\left( \frac{E_{tot}}{(E_{tot} - [Y])^2} + \frac{k_1+k_3}{k_5}+\frac{k_3}{k_4}\right)  <  \frac{k_7}
{k_6}\frac{k_2(k_1+k_3)}{k_4k_5}(2[Y]-E_{tot}).
 $$
The left-hand side of the inequality is an increasing function in $[Y]$ that tends to infinity as $[Y]$ approaches $E_{tot}$. 
The  right-hand side of the inequality is a line with positive slope that takes a negative value at $[Y]=0$ and crosses the $x$-axis at $[Y]=E_{tot}/2$.  If the line intersects the left-hand side curve, then there will be multiple steady states. The left-hand side does 
not depend on $k_2$ while the slope of the line increases with increasing  $k_2$. By fixing all constants except $k_2$ and letting $k_2$ vary arbitrarily, the two curves must meet. Except if they meet tangently, the two curves will cross in two points, between which $\varphi$ decreases. In this case, $\varphi$  increases initially, decreases for some interval, and increases 
to infinity afterwards. It follows that there are values of $S_{tot}$ for which the system has three steady states. 

Specific rate constants for which the system has three positive steady states are:
\begin{equation}\label{parameters}
k_i=1,\ i\neq 2,\quad k_2=2,\quad E_{tot}=10,\quad S_{tot}=100. 
\end{equation}
The three steady states correspond to $[Y]=6.5-\sqrt{11},\ 8,\ 6.5+\sqrt{11}$.

\medskip
{\bf Extension model 2. } We consider the canonical representative of the extension model class in the second column of Figure 3. The reactions of the model are those in \eqref{reactions} and \eqref{reactions2} together with
$$ S_0 +E \xrightarrow{k_6} Y_1,\qquad Y_1\xrightarrow{k_7} S_0+E,\qquad S_1 +E \xrightarrow{k_8} Y_2,\qquad 
Y_2\xrightarrow{k_9} S_1+E.$$
This model can be seen as an extension model of the canonical representative in column 1 (which is taken as the core model).  Canonical models always realize parameters of the core model (see Section~\ref{matching}).  Therefore, since  the canonical representative in column 1 admits multiple steady states, then so does  the canonical representative  in the second column of Figure 3 and it has at least 3 steady states. 

To show that it has at most 3 steady states we proceed as above. We consider the ODE system and we iteratively eliminate variables to obtain that at steady state
$$
 E =\frac{k_7k_9E_{tot}}{k_7k_9+k_6k_9[S_0] + k_7k_8[S_1]},
\qquad  [S_0]= \frac{  k_5k_7[S_1](k_9+k_8[S_1])}{k_9((k_1+k_3)k_7E_{tot} - k_5k_6[S_1])},$$
and 
$$S_{tot} = [S_0]+[S_1]+\left(\left(\frac{k_3}{k_4}+\frac{k_6}{k_7}\right)[S_0]  + \left(\frac{k_2}{k_4}+\frac{k_8}
{k_9}\right)[S_1]\right)[E].$$
By writing the expression above as a polynomial in $[S_1]$, we obtain a polynomial of degree 3 and hence at most three positive steady states can occur.

\medskip
{\bf Extension model 3. } We consider the canonical 
representative of the class of extension models   in the third column of Figure 3. 
The reactions of the model are those in \eqref{reactions} and \eqref{reactions2} together with
$$ S_1 +E \xrightarrow{k_6} Y_1,\qquad Y_1\xrightarrow{k_7} S_1+E, \qquad S_2 +E \xrightarrow{k_8} Y_2,\qquad 
Y_2\xrightarrow{k_{9}} S_2+E.$$
The ODEs of the   model  are the following:
\begin{align}
\dot{[S_0]} & =  -k_1[S_0][E] -k_3[S_0][E]  + k_5[S_1] ,  \label{eq:1b} \\
\dot{[S_1]} & =  -k_2[S_1][E] -k_5[S_1] - k_6[S_1][E] + k_1[S_0][E]+ k_4[S_2]+ k_7[Y_1],  \\
 \dot{[S_2]} & = -k_4[S_2] + k_2[S_1][E] - k_8[S_2][E] + k_3[S_0][E]+ k_9[Y_2],  \label{eq:3b}  \\ 
 \dot{[E]} & = -k_6[S_1][E] +k_7[Y_1] - k_8[S_2][E] + k_9[Y_2],\\
 \dot{[Y_1]} &=  k_6[S_1][E] - k_7[Y_1], \label{eq:6b} \\
\dot{[Y_2]} &=  k_8[S_2][E] - k_9[Y_2]. \label{eq:5b}
\end{align}
This system has two conservation laws:
$$S_{tot} = [S_0]+[S_1]+[S_2]+[Y_1]+[Y_2],\qquad E_{tot}=[E]+[Y_1]+[Y_2]. $$
From the steady-state equations derived from \eqref{eq:6b} and \eqref{eq:5b} we have that
$$ [Y_1] = \frac{k_6}{k_7} [S_1][E],\qquad  [Y_2] = \frac{k_8}{k_9} [S_2][E].$$
These expressions are increasing in $[S_1], [S_2]$, respectively.
Using the conservation law for $E_{tot}$ and the two expressions for $[Y_1],[Y_2]$, we obtain
$$ [E] =\frac{k_7k_9E_{tot}}{k_7k_9+k_6k_9[S_1] + k_7k_8[S_2]},$$
which is positive and decreasing in both $[S_1]$ and $[S_2]$.

Using the steady-state equation corresponding to \eqref{eq:1b} we have
$$[S_0]=  \frac{k_5[S_1]}{(k_1+k_3)[E]}, $$
which after substitution of $[E]$ with the expression for $[E]$ above, is increasing in $[S_1]$ and $[S_2]$.

We finally use the steady-state equation corresponding to \eqref{eq:3b}+\eqref{eq:5b} to obtain
$$k_4[S_2] =  \frac{k_2k_7k_9E_{tot}[S_1]}{k_7k_9+k_6k_9[S_1] + k_7k_8[S_2] }+ \frac{k_3k_5}{k_1+k_3}[S_1]. $$
Fix a value of $[S_1]$. The left-hand side of this equality is the line through the origin with slope $k_4$ in $[S_2]$. The right-hand side is a positive decreasing function of $[S_2]$ defined for all positive values of $[S_2]$. The function takes a positive value for $[S_2]=0$. It follows that the expressions on the two sides of the 
equality intersect in exactly one point for each fixed $[S_1]$. This is the steady-state value of $[S_2]$ corresponding to a given $[S_1]$.  Additionally, the left curve is independent of $[S_1]$ while the right curve increases in $[S_1]$. As a consequence, $[S_2]$ increases as a function of $[S_1]$.

Using the remaining conserved amount $S_{tot}$ we have that 
$$S_{tot} = [S_0]+[S_1]+[S_2]+[Y_1]+[Y_2],$$
where the right-hand side is expressed as an increasing positive function $\varphi$ of $[S_1]$, such that $\varphi(0)=0$ and such that it  tends to infinity as $[S_1]$ does. Hence, for every value $S_{tot}>0$ there exists a unique value $[S_1]>0$ satisfying $S_{tot}=\varphi(S_1)$. Using this value of $[S_1]$, all the other steady states concentrations are positive and can be found using the relations above.

We conclude that this model has exactly one positive steady state for all choices of rate constants and conserved amounts.

\end{document}